\documentclass[reqno,a4paper,oneside]{article}
\usepackage{graphics,amssymb,amsmath,setspace,tikz}
\usepackage[a4paper]{geometry}
\usepackage{epsfig,mathrsfs}
\usepackage[final]{pdfpages}
\usepackage{enumerate}
\usepackage{epstopdf}
\usepackage{rotating}
\usepackage{multirow}
\usepackage{graphicx}
\usepackage{hyperref}
\usepackage[style=numeric,sorting=none]{biblatex}
\usepackage[toc,page]{appendix}
\usepackage{anysize}
\marginsize{2.5cm}{2.5cm}{2.5cm}{2.5cm}

\usepackage{faktor}
\usepackage{cancel}

\usepackage{amsthm}

\usepackage{enumitem}

\usepackage{quiver}

\usepackage{authblk}

\newtheorem{theorem}{Theorem}[section]
\newtheorem{definition}{Definition}[section]
\newtheorem{lemma}{Lemma}[section]

\newtheorem{remark}{Remark}[section]

\numberwithin{equation}{section}

\title{Flat space limits of symmetric Taub-NUT instantons}
\author{Matthew Dales}
\affil{School of Computing and Mathematical Sciences, University of Leicester, University Road, Leicester, United Kingdom}

\begin{document}

\maketitle

\begin{abstract}
    We construct $U(1)$-symmetric $SU(2)$ instantons on the Taub-NUT space with generic holonomy, and instanton number and magnetic charge $(k,m)=(2,0)$, $(1,1)$, and $(2,-1)$. We construct these instantons via the bow construction, a Nahm-like non-linear transformation which identifies them with a set of bow data solving an integrable ODE and a series of algebraic relations on an interval. We show how the bow data recovers the data for all corresponding $U(1)$-symmetric instantons on flat $\mathbb{R}^4$ and calorons on $\mathbb{R}^3\times S^1$ in various limits, and we demonstrate the existence of Taub-NUT instantons that do not limit to $\mathbb{R}^4$ instantons nor calorons.
\end{abstract}

%\tableofcontents

\section{Introduction}\label{section: Introduction}

The Taub-NUT space is the prototypical example of an asymptotically locally flat (ALF) gravitational instanton, a Ricci-flat $4$-manifold asymptotic to a circle fibration over $S^2$, and interpolates between flat $\mathbb{R}^4$ and $\mathbb{R}^4\times S^1$ in various limits of the metric. It arises as a Riemannian solution to the vacuum Einstein equations with complexified time. There has been recent interest in the study of instantons (anti-self-dual connections) over gravitational instantons, and the Taub-NUT space in particular due to their connections to QCD, string theory, and the study of 4-manifolds \cite{Chalmers1997Three-dimensionalMonopoles,Hanany1997TypeDynamics,Witten2009BranesSpaces}. A complete construction of instantons on $\mathbb{R}^4$, called the ADHM construction, was presented in 1978 \cite{Atiyah1978ConstructionInstantons}, and generalised to monopoles on $\mathbb{R}^3$ by Nahm in 1983 \cite{Nahm1983AllGroups}. Both of these constructions, and their more general counterparts, involve a non-linear transformation between the anti-self-dual connection and a set of data defined over a lower dimensional space. In the case of the ADHM construction this is to a single point, at which we have a set of matrices solving a series of purely algebraic relations, while in the Nahm transform we have an interval over which we have data solving an integrable ODE. These Nahm-like transformations establish an isometry between the moduli space of anti-self-dual connections and the moduli space of Nahm data (which typically amounts to a simpler problem). This idea has been extended to calorons, anti-self-dual connections on $\mathbb{R}^3\times S^1$ \cite{Charbonneau2010TheCalorons}, instantons on ALE spaces \cite{Kronheimer1990Yang-MillsInstantons}, and recently instantons on ALF spaces \cite{Cherkis2021InstantonsTheorem,Cherkis2024InstantonsConstruction,Cherkis2025InstantonsIsometry}.

Taub-NUT instantons are classified by 3 invariants: an instanton number $k$, a monopole charge $m$, and a holonomy parameter $\lambda$. Instantons as direct solutions to the anti-self-duality equation have only been found in the case of trivial asymptotic holonomy \cite{Etesi2001GeometricSpace}, in which there is a full $U(2)$ symmetry. By considering the Nahm-like transform for Taub-NUT instantons, which we will refer to as the `bow transform' we can construct the moduli spaces for instantons with non-trivial asymptotic holonomy, higher charge, and smaller symmetry group. The moduli space of all instantons with generic asymptotic holonomy, instanton number $1$ and monopole charge $0$ has been constructed \cite{Cherkis2009ModuliSpace}, and is the only completely classified case. Studying symmetric examples has proven productive for other, similar objects such as: monopoles \cite{Hitchin1995SymmetricMonopoles,Houghton1996TetrahedralMonopoles,Houghton1996OctahedralMonopoles,Charbonneau2022ConstructionSymmetries}, instantons \cite{Singer1999SymmetricFields,Cork2022ADHMSkyrmions}, and calorons \cite{Ward2004SymmetricCalorons,Harland2007LargeCalorons,Cork2018SymmetricMap,Kato2021SymmetricLimits}. By considering a group action on the space of solutions that acts as an isometry we can construct a geodesic submanifold of the moduli space, this is useful for studying the low-energy dynamics of instantons, since the trajectories will approximately follow geodesics in the moduli space \cite{Manton1982AMonopoles,Allen2013TheInstantons}. The complete moduli space of Taub-NUT instantons is hyperk\"{a}hler \cite{Cherkis2024InstantonsConstruction}, with a small enough symmetry group it is possible for one of the complex structures to be preserved, leaving a K\"{a}hler submanifold. We will extend the set of classified instantons to all $U(1)$-symmetric Taub-NUT instantons with generic asymptotic holonomy, and $(k,m)=(2,0),(1,1),(2,-1)$, ie $\max(k,k+m)=2$. In particular, the solutions we construct are analytic, and are therefore useful for studying the properties of the moduli space, including its metric.

The Taub-NUT space interpolates between the flat $\mathbb{R}^4$ and $\mathbb{R}^3\times S^1$ in various limits of the metric; it is therefore natural to ask whether the Taub-NUT instanton descends to a flat $\mathbb{R}^4$ instanton and a caloron in these limits. We will demonstrate that $U(1)$-symmetric $\mathbb{R}^4$ instantons and calorons are recovered in these limits, in particular all calorons found by Harland \cite{Harland2007LargeCalorons} arise this way, and that there exist Taub-NUT instantons that do not limit to calorons.

Sections \ref{section: Taub-NUT instantons} and \ref{section: The bow transform} review the definitions of the Taub-NUT space, Taub-NUT instantons, the bow, and the bow data; section \ref{section: The bow transform} also outlines the bow transform between Taub-NUT instantons and the bow data. In sections \ref{section: U(1) invariant (2,2,2)}, \ref{section: U(1) invariant solutions (1,2,1)}, and \ref{section: U(1) invariant solutions of type (2,1,2)} we will construct $U(1)$ symmetric instantons of charges $(k,m)=(2,0),(1,1),(2,-1)$ respectively under specific actions, and in section \ref{section: classification of solutions} we will prove that all $U(1)$ symmetric bow data can be constructed in these ways. Finally, in section \ref{section: limits of bow data} we will consider the instanton and caloron limits of the bow data.

% In section \ref{section: Taub-NUT instantons} we will introduce the Taub-NUT space and Taub-NUT instantons, and in section \ref{section: The bow transform} we will set-up the bow data and briefly describe the bow transform between Taub-NUT instantons and bow data. In sections \ref{section: U(1) invariant (2,2,2)}, \ref{section: U(1) invariant solutions (1,2,1)}, and \ref{section: U(1) invariant solutions of type (2,1,2)} we will construct $U(1)$ symmetric instantons of charge $(k,m)=(2,0),(1,1),(2,-1)$ respectively under a specific action, and in section \ref{section: classification of solutions} we will prove a classification of all $U(1)$ symmetric bow data by the constructed solutions. Finally in section \ref{section: limits of bow data} we will consider the instanton and caloron limits of the bow data.

\section{Taub-NUT instantons}\label{section: Taub-NUT instantons}

\subsection{The Taub-NUT space}\label{subsection: The Taub-NUT space}

The Taub-NUT space $\Sigma_1$ is an oriented Riemannian 4-manifold, which can be viewed as a circle bundle over $\mathbb{R}^3$ with a Dirac monopole at $\vec{q}\in\mathbb{R}^3$; in principle $\vec{q}$ can be moved to the origin by a coordinate choice, however we will keep it as a parameter to study the limit $|\vec{q}|\rightarrow\infty$. Let $\vec{t}\in\mathbb{R}^3$ and $\tau\in S^1_{4\pi}$ be coordinates on $\Sigma_1$, then the Taub-NUT metric takes the Gibbons-Hawking form \cite{Gibbons1979ClassificationSymmetries}
\begin{equation}\label{eq:Taub-NUT metric}
    ds^2=\frac{1}{4}\left(V(\vec{t})d\vec{t}^2+\frac{1}{V(\vec{t})}(d\tau+\omega)^2\right)\qquad:\qquad V(\vec{t})=l+\frac{1}{|\vec{t}-\vec{q}|}
\end{equation}
where $l\in\mathbb{R}_+$, and $\omega\in\Lambda^1(\mathbb{R}^3\setminus\lbrace\vec{q}\rbrace)$ satisfies $d\omega=\star_{(3)}dV$, where $\star_{(3)}$ denotes the Hodge star over Euclidean $\mathbb{R}^3$.

The Taub-NUT space is an asymptotically locally flat (ALF) gravitational instanton, ie it is Ricci-flat and asymptotically a circle fibration $S^1\hookrightarrow S^3\rightarrow S^2$, in this case it is degree 1, and therefore the Hopf fibration.

\subsection{Taub-NUT instantons}\label{subsection: Taub-NUT instantons}

Let $E\rightarrow\Sigma_1$ be a rank-2 hermitian vector bundle over $\Sigma_1$, then define
\begin{definition}
    A \textbf{Taub-NUT instanton} is a connection $A$ on $E\rightarrow\Sigma$ with anti-self-dual curvature $F=dA+A\wedge A$, and finite Yang-Mills energy
    \begin{equation}\label{eq:Yang-Mills energy}
        E_{\operatorname{YM}}[A]:=-\int_{\Sigma_1}\operatorname{tr}(F\wedge \star F).
    \end{equation}
\end{definition}
The finite energy condition implies the existence of 3 invariants which classify the instanton, the instanton number $k\in\mathbb{Z}$, the monopole charge $m\in\mathbb{Z}$, and the holonomy parameter $\lambda\in\mathbb{R}$; we will briefly describe these following the construction in \cite{Cherkis2011InstantonsGravitons}.

Consider $S^2_R\subset\mathbb{R}^3$, for some large $R$; at each point $p\in S^2_R$ we have a $\tau$-parameterised circle, so we have a 3-sphere $S^3_R$ fibred by the $\tau$-circles. For the Taub-NUT space this is the Hopf-fibration $S^1\hookrightarrow S^3_R\rightarrow S^2_R$. Since the instanton connection $A$ has finite energy, there exists a gauge in which $A|_{S^3_R}$ has components that are $\tau$-independent, and so take the form
\begin{equation}\label{eq:Connection at R}
    A|_{S^3_R}=\hat{A}+\hat{\Phi}\frac{d\tau+\omega}{V}
\end{equation}
where $\hat{A}$ is a connection of $\mathbb{R}^3\setminus\lbrace\vec{q}\rbrace$, and $\hat{\Phi}:\mathbb{R}^3\setminus\lbrace\vec{q}\rbrace\rightarrow\mathfrak{su}(2)$ is a Lie algebra valued function. The anti-self-duality condition for $A$ restricted to $S_R^3$ is then equivalent to the Bogomolny equation for $(\hat{A},\hat{\Phi})$
\begin{equation}\label{eq:Bogomolny}
    \hat{F}+\star_{(3)}\hat{D}\hat{\Phi}=0.
\end{equation}
The above gauge choice requires a map $g:S^3_R\rightarrow SU(2)\cong S^3$ from the 3-sphere onto itself, $g$ is therefore characterised by an element of $\pi^3(S^3)=\mathbb{Z}$, which is the instanton number $k$. Let $W\rightarrow\mathbb{R}^3$ be an $SU(2)$-bundle over $\mathbb{R}^3$, with restriction $W_R:=W|_{S^2_R}$, with connection $\hat{A}_R$. Let $E_\infty,W_\infty$ be the restrictions of $E$ and $W$ to $S^3_\infty$ and $S^2_\infty$ respectively. By the Birkhoff-Grothendieck theorem $W_\infty$ decomposes as the direct sum of 2 line bundles
\begin{equation}\label{eq:W decomposition}
    W_\infty=L_+\oplus L_-.
\end{equation}
The first Chern numbers of $L_\pm$ must sum to zero, since $W$ is trivial, so we denote $m=c_1(L_+)$, which is the monopole charge of the instanton. $W_\infty$ has connection $\hat{A}_\infty$, which is diagonal with respect to the splitting
\begin{equation*}
    \hat{A}_\infty=\hat{A}_++\hat{A}_-
\end{equation*}
where $\hat{A}_\pm$ is the unique connection on $L_\pm$ with constant curvature with respect to the volume form. The section $\hat{\Phi}_\infty\in\Gamma(\operatorname{End}(W_\infty))$ has eigenbundles $L_\pm$, with eigenvalues $\pm i\lambda$, for $\lambda\in[0,\frac{l}{2}]$. With respect to a local gauge the connection $A$ therefore takes the asymptotic form:
\begin{equation}\label{eq:Asymptotic form of the connection}
    A\sim\left(\left(\lambda+\frac{m}{2|\vec{t}-\vec{q}|}\right)\frac{d\tau+\omega}{V}\right)\begin{pmatrix}
        -i &0 \\ 0 & i
    \end{pmatrix}+\begin{pmatrix}
        \hat{A}_- & 0 \\ 0 & \hat{A}_+
    \end{pmatrix}+\mathcal{O}(|\vec{t}-\vec{q}|^2).
\end{equation}

Let $\mathcal{C}_{k,m}(\lambda)$ denote the space of connections that solve the anti-self-duality equations on $\Sigma_1$, and the appropriate boundary conditions for an instanton with instanton number $k$, monopole charge $m$, and holonomy parameter $\lambda$. We consider two elements in $\mathcal{C}_{k,m}(\lambda)$ to represent the same instanton if they differ only by a gauge transformation $g:\Sigma_1\rightarrow SU(2)$, such that $g|_{S^3_\infty}=\operatorname{id}$, so we are therefore interested in the moduli space of framed solutions
\begin{equation}\label{eq:Moduli space}
    \mathcal{M}_{k,m}(\lambda):=\faktor{\mathcal{C}_{k,m}(\lambda)}{\mathcal{G}}
\end{equation}
where $\mathcal{G}$ is the space of all possible gauge transformations.

Solutions to the anti-self-duality equation over $\Sigma_1$ have been constructed directly, notably by Etesi and Hausel \cite{Etesi2001GeometricSpace}. By considering a harmonic function ansatz similar to that of the JNR instanton on $\mathbb{R}^4$ \cite{Jackiw1977ConformalConfigurations}, a 1-parameter family of $SU(2)$ Taub-NUT instantons is constructed which have a $U(2)$ symmetry and trivial holonomy parameter $\lambda=0$. A construction similar to the Nahm transform for monopoles \cite{Nahm1983AllGroups}, or the ADHM construction of instantons on $\mathbb{R}^4$ \cite{Atiyah1978ConstructionInstantons}, has been developed for Taub-NUT instantons, and has been proven to be complete \cite{Cherkis2021InstantonsTheorem,Cherkis2024InstantonsConstruction,Cherkis2025InstantonsIsometry}. Such a construction allows us to compute the moduli spaces of Taub-NUT instantons for higher charges, lower symmetry, and generic asymptotic holonomy.

\section{The bow transform}\label{section: The bow transform}

\subsection{Bow data}\label{subsection: Bow data}

Let $\mathcal{I}=[-\frac{l}{2},\frac{l}{2}]$ be an interval with coordinate $s$. We designate the points $\pm\frac{l}{2}$ as bifundamental points, and the points $\pm\lambda$ as fundamental points. These points divide the interval into 3 subintervals $\mathcal{I}_1=[-\frac{l}{2},-\lambda]$, $\mathcal{I}_2=[-\lambda,\lambda]$, and $\mathcal{I}_3=[\lambda,\frac{l}{2}]$, the lengths of which we denote $l_1,l_2,l_3\in[0,l]$. Over each interval we have a hermitian bundle $\mathcal{E}_P\rightarrow\mathcal{I}_P$, of rank $\mathcal{R}_P$, such that $\mathcal{R}_1=\mathcal{R}_3=k$, and $\mathcal{R}_2=k+m$. If $m=0$, we have a constant rank bow, and to each fundamental point we associate a 1-complex dimensional vector space $V_\pm$. To a given bow we associate a set of bow data $\mathcal{BD}=(\mathcal{T},\mathcal{B},\mathcal{Q})$, a tuple consisting of 3 components: the Nahm data $\mathcal{T}$, the bifundamental data $\mathcal{B}$, and the fundamental data $\mathcal{Q}$. Here we will only consider $m=0$ and $m=\pm1$, the general set-up is presented in \cite{Cherkis2011InstantonsGravitons}. We will refer to bows as type-$(\mathcal{R}_1,\mathcal{R}_2,\mathcal{R}_3)$ to denote the ranks of the bundles.

\subsubsection*{The Nahm data}\label{subsubsection: The Nahm data}

The Nahm data $\mathcal{N}$ consists of a connection $\nabla_P=\frac{d}{ds}+T_P^0$ on the bundle $\mathcal{E}_P$, and a triple of bundle endomorphisms $T_P^j:\mathcal{I}_P\rightarrow\mathfrak{u}(\mathcal{R}_P)$ solving Nahm's equations
\begin{equation}\label{eq:Nahm's equations}
    \frac{dT_P^j}{ds}+[T_P^0,T_P^j]+\frac{1}{2}\epsilon_{jkl}[T_P^k,T_P^l]=0
\end{equation}
on the interior of each interval. We require that the Nahm data are continuous over each subinterval but not over the whole interval. For the most part we will be working with Donaldson's Nahm complexes \cite{Donaldson1984NahmsMonopoles}:
\begin{equation}\label{eq:Nahm complexes}
    \alpha_P:=T_P^0+iT_P^1\qquad;\qquad \beta_P:=T_P^2+iT_P^3
\end{equation}
in terms of which Nahm's equations become:
\begin{align}
    \frac{d(\alpha_P+\alpha_P^\dagger)}{ds}+[\alpha_P,\alpha_P^\dagger]+[\beta_P,\beta_P^\dagger]=0\label{eq:Real Nahm equation}\\
    \frac{d\beta_P}{ds}+[\alpha_P,\beta_P]=0\label{eq:Complex Nahm equation}
\end{align}
and from which we construct
\begin{equation}\label{eq:AP matrix}
    A_P(s;\xi):=\beta_P(s)+(\alpha_P(s)+\alpha_P^\dagger(s))\xi-\beta_P^\dagger(s)\xi^2\qquad:\qquad \xi\in\mathbb{C}
\end{equation}
in terms of which we will write the matching conditions at the marked points.

\subsubsection*{The bifundamental data}\label{subsubsection: The bifundamental data}

The bifundamental data $\mathcal{B}$ is a pair $B_\pm\in\mathbb{C}^{k\times k}$ of homomorphisms between the fibres of $\mathcal{E}_1$ at $-\frac{l}{2}$ and $\mathcal{E}_3$ at $\frac{l}{2}$, we require that $(B_-,B_+)\neq0$, and that the bifundamental matching conditions are satisfied:
\begin{align}
    A_1(-\frac{l}{2})=(B_+-B_-^\dagger\xi)(B_-+B_+^\dagger\xi)+q\label{eq:Bifundamental matching condition -}\\
    A_3(\frac{l}{2})=(B_-+B_+^\dagger\xi)(B_+-B_-^\dagger\xi)+q\label{eq:Bifundamental matching condition +}
\end{align}
for all $\xi\in\mathbb{C}$, where $q=(q_3-iq_2+2q_1\xi-(q_3+iq_2)\xi^2)\otimes\operatorname{id}_k$ for $q_1,q_2,q_3\in\mathbb{R}$ the coordinates of the Taub-NUT centre.

\subsubsection*{The fundamental data}\label{subsubsection: The fundamental data}

The form of the fundamental data $\mathcal{Q}$ is dependent on the value of $m$.

For $m=0$, the fundamental data consists of homomorphisms $I_\pm,J_\pm^\dagger\in\mathbb{C}^k$ between the fibres of $\mathcal{E}_P$ at $\pm\lambda$ and the vector spaces $V_\pm$, such that $(I_\pm,J^\dagger_\pm),(I_-,J^\dagger_-),(I_+,J_+^\dagger)\neq0$, and
\begin{align}
    A_1(-\lambda)-A_2(-\lambda)=(I_--J_-^\dagger\xi)(J_-+I_-^\dagger\xi)\label{eq:Fundamental matching condition -}\\
    A_2(\lambda)-A_3(\lambda)=(I_+-J_+^\dagger\xi)(J_++I_+^\dagger\xi)\label{eq:Fundamental matching condition +}
\end{align}
for all $\xi\in\mathbb{C}$.

For $m=1$, the Nahm data changes rank across the fundamental points, with rank $k$ on $\mathcal{I}_{1,3}$ and $k+1$ on $\mathcal{I}_2$. Now, the fundamental data consists of homomorphisms $X,Y\in\mathbb{C}^{(k+1)\times k}$ between the fibres of $\mathcal{E}_P$ on either side of the fundamental point, such that $X^\dagger X=Y^\dagger Y=\operatorname{id}_k$, and
\begin{equation}\label{eq:Fundamental matching condition non-zero magnetic charge}
    X^\dagger T_2^j(-\lambda)X=T_1^j(-\lambda)\qquad;\qquad Y^\dagger T_2^j(\lambda) Y=T_3^j(\lambda).
\end{equation}

% For $m=-1$, the outer intervals are higher rank, so we now have fundamental data $X,Y\in\mathbb{C}^{k\times(k+1)}$ such that $XX^\dagger=YY^\dagger=\operatorname{id}_k$, with matching condition
% \begin{equation}
%     X^\dagger T_2^j(-\lambda)X=T_1^j(-\lambda)\qquad;\qquad Y^\dagger T_2^j(\lambda)Y=T_3^j(\lambda).
% \end{equation}

For $m=-1$, the outer intervals have the higher rank, so the fundamental data is $W,Z\in\mathbb{C}^{(k+1)\times k}$ such that $W^\dagger W=Z^\dagger Z=\operatorname{id}_k$, and such that
\begin{equation}
    W^\dagger T_1^j(-\lambda)W=T_2^j(-\lambda)\qquad;\qquad Z^\dagger T_3^j(\lambda)Z=T_2^j(\lambda).
\end{equation}

\subsubsection*{Gauge transformations}\label{subsubsection: Gauge transformations}

Gauge transformations on each subinterval $g_P:\mathcal{I}_P\rightarrow U(\mathcal{R}_P)$ must be smooth on the interior of the interval, and for $m=0$ we require $g_1(-\lambda)=g_2(-\lambda)$ and $g_2(\lambda)=g_3(\lambda)$. Gauge transformations act on the bow data as:
\begin{equation}\label{eq:Gauge transformations on bow data}
    g:\begin{pmatrix}
        T_P^0 \\ T_P^j \\ \\ B_+ \\ B_- \\ \\ I_\pm \\ J_\pm \\ \\ X \\ Y \\ \\ W \\ Z
    \end{pmatrix}\rightarrow \begin{pmatrix}
        g_P^{-1}T_P^0g_P+g_P^{-1}\frac{d}{ds}g_P \\ g_P^{-1}T_P^j g_P \\ \\ g_1^{-1}(-\frac{l}{2})B_+g_3(\frac{l}{2}) \\ g_3^{-1}(\frac{l}{2})B_-g_1(-\frac{l}{2}) \\ \\ g_2^{-1}(\pm\lambda)I_\pm \\ J_\pm g_2(\pm\lambda) \\ \\ g_2^{-1}(-\lambda)Xg_1(-\lambda) \\ g_2^{-1}(\lambda)Yg_3(\lambda) \\ \\ g_1^{-1}(-\lambda)Wg_2(-\lambda) \\ g_3^{-1}(\lambda)Zg_2(\lambda)
    \end{pmatrix}.
\end{equation}

It is always possible to choose a gauge in which $T_P^0=0$ identically across all intervals, unless otherwise specified this is the gauge in which we will work; in terms of the complex Nahm data this is equivalent to fixing $\alpha_P$ to be hermitian on all intervals. The set of remaining gauge transformations is $\mathcal{G}_0$ comprised of constant transformations on each subinterval.

We are principally interested in the moduli space of bow data. Consider a bow with instanton number $k$, monopole charge $m$, and fundamental points at $\pm\lambda$. We denote by $\tilde{\mathcal{C}}_{k,m}(\lambda)$ the space of all bow data solving the bow conditions over such a bow, and $\tilde{\mathcal{G}}$ the space of all gauge transformations. The moduli space is
\begin{equation}\label{eq:Moduli space of bow data}
    \tilde{\mathcal{M}}_{k,m}(\lambda)=\faktor{\tilde{\mathcal{C}}_{k,m}(\lambda)}{\tilde{\mathcal{G}}}.
\end{equation}
The moduli space $\tilde{\mathcal{M}}_{k,m}(\lambda)$ of bow data is diffeomorphic to the moduli space $\mathcal{M}_{k,m}(\lambda)$ of Taub-NUT instantons \cite{Cherkis2021InstantonsTheorem,Cherkis2024InstantonsConstruction,Cherkis2025InstantonsIsometry}. There is a natural hyperk\"{a}hler $L^2$ metric on each moduli space, and the Bow transform, which we will summarise later, is a hyperk\"{a}hler isometry between $\mathcal{M}_{k,m}(\lambda)$ and $\tilde{\mathcal{M}}_{k,m}(\lambda)$.

\subsection{Abelian instantons}\label{subsection: Abelian instantons}

The Taub-NUT space itself is the moduli space of a particularly simple bow \cite{Cherkis2011InstantonsGravitons,Gibbons1997HyperKahlerSpaces}, in which there are no fundamental points and the bow has constant rank $k=1$. In this case the Taub-NUT space is the $U(1)$ gauge quotient manifold of a level set $\mu^{-1}(\vec{q})$, where $\mu$ is the hyperk\"{a}hler moment map of the gauge group action, therefore $\mu^{-1}(\vec{q})$ is a principal bundle over $\Sigma_1$. For any point $s\in\mathcal{I}$, consider the fibre $\mathcal{E}_s$ of the hermitian bundle over the bow at that point, and the trivial bundle $\mathcal{E}_s\times\mu^{-1}(\vec{q})\rightarrow\mu^{-1}(\vec{q})$. Since the $U(1)$ gauge group acts on both $\mathcal{E}_s$ and $\mu^{-1}(\vec{q})$, we obtain an associated line bundle $e_s\rightarrow\Sigma_1$ for each $s\in\mathcal{I}$. We therefore have a bow-parameterised family of line bundles $e_s$, each of which carries an abelian instanton connection $a_s$, ie a $U(1)$ connection with anti-self-dual curvature and finite action. In an appropriate choice of coordinates and gauge, the abelian instanton takes the form
\begin{equation}\label{eq:Abelian instanton}
    a_s=is\frac{d\tau+\omega}{V}.
\end{equation}

\subsection{The up and down transforms}\label{subsection: The up and down transforms}

The isometry between the moduli space $\tilde{\mathcal{M}}_{k,m}(\lambda)$ and of Taub-NUT instantons $\mathcal{M}_{k,m}(\lambda)$ is a generalisation of the ADHM-Nahm transforms for instantons on $\mathbb{R}^4$ \cite{Atiyah1978ConstructionInstantons}, monopoles \cite{Nahm1983AllGroups}, and calorons \cite{Charbonneau2010TheCalorons}. This consists of an `up transform' which takes a set of bow data and produces an instanton, and a `down transform' which takes an instanton and produces bow data \cite{Cherkis2021InstantonsTheorem,Cherkis2024InstantonsConstruction,Cherkis2025InstantonsIsometry}. We will briefly summarise these transforms in the following sections.

\subsubsection*{The up transform}\label{subsubsection: the up transform}

Consider a set of bow data $(T_P^\alpha,B_\pm,I_\pm,J_\pm)$, let $e^\alpha$ be a 2-dimensional representation of the quaternions, and define $\mathbf{T}=\sum_\alpha T^\alpha\otimes e^\alpha$, and $\mathbf{t}=\tau\otimes e^0+\sum_j t^j\otimes e^j$ for $(\tau,\vec{t})\in\Sigma_1$. Then, on each interval $\mathcal{I}_P$, let $U_p(s;\mathbf{t})$ be $\mathcal{R}_P$-component column vectors of complex quaternions satisfying
\begin{equation}\label{eq:Differential equation for U_P}
    \frac{d}{ds}U_P=\left(\mathbf{T}-i\mathbf{t}\right)U_P
\end{equation}
on the interior of each interval. At $s=\pm\frac{l}{2}$, we have complex $k$-column vectors $Z_\pm$ such that
\begin{align}
    U_3\left(\frac{l}{2}\right)+Z_+B^\dagger=0\qquad&:\qquad B=\frac{1}{\sqrt{2}}\begin{pmatrix}
        B_+^\dagger+B_- \\ B_+^\dagger-B_-
    \end{pmatrix}\label{eq:Bifundamental matching condition for U_3}\\
    U_1\left(-\frac{l}{2}\right)-Z_-B_c^\dagger=0\qquad&:\qquad B_c=\frac{1}{\sqrt{2}}\begin{pmatrix}
        B_-^\dagger-B_+\\ B_-^\dagger+B_+
    \end{pmatrix}.\label{eq:Bifundamental matching condition for U_1}
\end{align}
For $m=0$, we have $W_\pm\in\mathbb{C}$ satisfying
\begin{equation}\label{eq:Fundamental matching conditions for U_P}
    \begin{split}
        U_2(\lambda)-U_3(\lambda)+W_+Q_+^\dagger=0 \\ U_1(-\lambda)-U_2(-\lambda)+W_-Q_-^\dagger=0
    \end{split}\qquad:\qquad Q_\pm=\frac{1}{\sqrt{2}}\begin{pmatrix}
        J_\pm^\dagger-I_\pm \\ J_\pm^\dagger+I_\pm
    \end{pmatrix}.
\end{equation}
We require that $U_P,Z_\pm,W_\pm$ satisfy the normalisation condition:
\begin{equation}\label{eq:Normalisation condition for bow construction data}
    \sum_{P=1}^3\int_{\mathcal{I}_P}U_P^\dagger U_Pds+\sum_{\pm}Z_\pm^\dagger Z_\pm+W_\pm^\dagger W_\pm=e^0.
\end{equation}
Then, the components of the Taub-NUT instanton connection are given by:
\begin{equation}\label{eq:Construction of taub-nut instanton from bow data}
    A_\alpha=\sum_{P=1}^3\int_{\mathcal{I}_P}U_P^\dagger\frac{\partial}{\partial t^\alpha}U_Pds+\sum_{\pm}\left(Z_\pm^\dagger\frac{\partial}{\partial t^\alpha}Z_\pm+W_\pm^\dagger\frac{\partial}{\partial t^\alpha}W_\pm\right)
\end{equation}
where $t^0=\tau$.

For $m=1$, we do not define $W_\pm$, and instead require that $U_P$ satisfy
\begin{equation}\label{eq:matching condition for U_P for non-zero magnetic charge}
    X^\dagger U_2(-\lambda)=U_1(-\lambda)\qquad,\qquad Y^\dagger U_2(\lambda)=U_3(\lambda)
\end{equation}
the normalisation condition and the construction of $A_\alpha$ are the same but without terms involving $W_\pm$; and we have a similar construction for $m=-1$, accounting for the outer intervals now having the higher rank.

\subsubsection*{The down transform}\label{subsubsection: the down transform}

Let $A$ be an $SU(2)$ instanton connection on a bundle $E\rightarrow\Sigma_1$, and let $a_s$ be a bow parameterised family of abelian instanton connections on $U(1)$ bundles $e_s\rightarrow\Sigma_1$. For each $s\in\mathcal{I}$ we can associate to $E$ a line bundle $E\otimes e_s$, each equipped with an instanton connection $\nabla^s$ given by the induced connection. We choose an orientation of $\Sigma_1$ with volume form $\operatorname{Vol}=Vd\tau\wedge dt_1\wedge dt_2\wedge dt_3$, and the orthonormal frame
\begin{equation}\label{eq:Orthonormal frame of the tangent bundle}
    \Theta_0=\sqrt{V}\partial_\tau\quad,\quad \Theta_j=\frac{1}{\sqrt{V}}(\partial_j-\omega_j\partial_\tau)
\end{equation}
for the tangent bundle, and
\begin{equation}\label{eq:Dual coframe}
    \theta^0=\frac{1}{\sqrt{V}}(d\tau+\omega)\quad,\quad \theta^j=\sqrt{V}dt^j
\end{equation}
its dual coframe \cite{Cherkis2025InstantonsIsometry}. Let $\nabla_X^s$ be the $X$-direction covariant derivative of $\nabla^s$. We note that the covariant derivate is independent of $s$ for any $X\perp\Theta_4$, ie $\nabla_X^{s_1}=\nabla_X^{s_2}$ for any $s_1,s_2\in\mathcal{I}$. Then we can define the Dirac operator $\mathfrak{D}_s$ and its formal adjoint $\mathfrak{D}_s^\dagger$:
\begin{equation}
    \mathfrak{D}_s:=\nabla_{\Theta_0}^s+i\sigma^j\nabla_{\Theta_j}^s\qquad;\qquad \mathfrak{D}_s^\dagger:=-\nabla_{\Theta_0}^s+i\sigma^j\nabla_{\Theta_j}^s.
\end{equation}
The Dirac operator can be written as
\begin{equation}
    \mathfrak{D}_s=\mathfrak{D}_0+a_s\qquad:\qquad \mathfrak{D}_0=\nabla_{\Theta_0}^0+i\sigma^j\nabla_{\Theta_j}
\end{equation}
where $\nabla_{\Theta_0}^0(\Phi\otimes\phi_s)=(D_{\Theta_0}^A\Phi)\otimes\phi_s+\Phi\otimes(\partial_{\Theta_0}\phi_s)$ is independent of $s$, acting on $\Phi\in\Gamma(E),\phi_s\in\Gamma(e_s)$.

Consider $\psi_s^i\in\operatorname{ker}(\mathfrak{D}_s)$ for $i=1,2,\dots,d=:\operatorname{dim}(\operatorname{ker}(\mathfrak{D}_s))$ such that
\begin{equation}\label{eq:orthonormal condition on kernel of Dirac operator}
    \int_{\Sigma_1}(\psi_s^i)^\dagger\psi_s^jd\operatorname{Vol}=\delta^{ij}.
\end{equation}
Let $\Psi_s:=\begin{pmatrix}
    \psi_s^1 & \psi_s^2& \cdots & \psi_s^d
\end{pmatrix}$ then the Nahm data on the bow $\mathcal{I}$ are given by:
\begin{equation}\label{eq:Construction of Nahm data}
    T^j(s)=\int_{\Sigma_1}t^j\Psi_s^\dagger\Psi_s d\operatorname{Vol}\quad,\quad T^0(s)=\int_{\Sigma_1}\Psi_s^\dagger\frac{d}{ds}\Psi_s d\operatorname{Vol}.
\end{equation}
Details on the construction of the fundamental and bifundamental matching data are to be found in \cite{Cherkis2025InstantonsIsometry}.

\subsubsection*{The reality condition}\label{subsubsection: the reality condition}

\begin{theorem}[The reality condition]\label{The reality condition}
    The bow data of an $SU(2)$ instanton on $\Sigma_1$ is always gauge equivalent to a set of bow data satisfying the additional constraints:
    \begin{equation}\label{eq:The reality condition - equation}
        T^j(-s)=T^j(s)^t
    \end{equation}
    for all $s\in\mathcal{I}$, and
    \begin{align}
        B_+B_-=(B_-B_+)^t \quad,\quad B_+B_+^\dagger-B_-^\dagger B_-=B_+^\dagger B_+-B_-B_-^\dagger\label{eq:Reality condition on bifundamental data}\\
        I_-J_-=-(I_+J_+)^t\quad,\quad I_-I_-^\dagger-J_-J_-^\dagger=-(I_+I_+^\dagger-J_+^\dagger J_+)^t\quad,\quad Y=\bar{X}\label{eq:reality condition for fundamental data}.
    \end{align}
\end{theorem}

\begin{proof}
    Since $Sp(1)\cong SU(2)$, we have that $A$ must satisfy
    \begin{equation}\label{eq:Sp(1) condition on A}
        A=\Omega^{-1}A^*\Omega\quad:\quad \Omega=i\sigma^2.
    \end{equation}
    Then, considering the form of $a_s$ \eqref{eq:Abelian instanton}, \eqref{eq:Sp(1) condition on A} implies that
    \begin{equation}
        \mathfrak{D}_s=\Omega^{-1}\mathfrak{D}_{-s}^*\Omega.
    \end{equation}
    Now, let $\psi_s^j\in\operatorname{ker}(\mathfrak{D}_s)$, and form $\Psi_s=\begin{pmatrix}
        \psi_s^1&\psi_s^2&\cdots&\psi_s^d
    \end{pmatrix}$, then:
    \begin{equation}
        \mathfrak{D}_s\Psi_s=0\quad\Rightarrow\quad(\Omega^{-1}\mathfrak{D}_{-s}^*\Omega)\Psi_s=(\Omega\mathfrak{D}_{-s}\Omega^{-1}\Psi^*)^*=0
    \end{equation}
    so, $\Omega^{-1}\Psi_s^*\in\operatorname{ker}(\mathfrak{D}_{-s})$. Therefore, there exists a gauge in which $\Psi_{-s}=\Omega^{-1}\Psi_s^*$, and so, considering the construction of the Nahm matrices:
    \begin{equation}
        T^j(-s)=\int_{\Sigma_1}t^j\Psi_{-s}^\dagger\Psi_{-s}d\operatorname{Vol}=\left(\int_{\Sigma_1}t^j\Psi_s^\dagger\Psi_sd\operatorname{Vol}\right)^t
        =T(s)^t.
    \end{equation}
    Considering Nahm data that satisfy $T^j(-s)=T^j(s)^t$ and therefore $A(-s;\xi)=A(s;\xi)^t$ \eqref{eq:AP matrix}, and considering the matching conditions \eqref{eq:Bifundamental matching condition -}, \eqref{eq:Bifundamental matching condition +}, \eqref{eq:Fundamental matching condition -}, \eqref{eq:Fundamental matching condition +}, \eqref{eq:Fundamental matching condition non-zero magnetic charge}, we can see that
    \begin{align}
        &\left((B_+-B_-^\dagger\xi)(B_-+B_+^\dagger\xi)\right)^t=(B_-+B_+^\dagger\xi)(B_+-B_-^\dagger\xi)\\
        &\left(I_--J_-^\dagger\xi)(J_-+I_-^\dagger\xi)\right)^t=-(I_+-J_+^\dagger\xi)(J_++I_+^\dagger\xi)\\
        &X^\dagger T_2^j(-\lambda)X=Y^tT_2^j(-\lambda)\bar{Y}.
    \end{align}
    Expanding in orders of $\xi$, we have \eqref{eq:reality condition for fundamental data} and \eqref{eq:Reality condition on bifundamental data}, as required.
\end{proof}

\section{U(1) invariant solutions of type (2,2,2)}\label{section: U(1) invariant (2,2,2)}

We are interested in finding bow data symmetric under an action $\varphi:G\times\mathcal{BD}\rightarrow\mathcal{BD}$, which, under the bow transform, produces an instanton on the Taub-NUT space symmetric under a corresponding action. We therefore wish to find a set of bow data solving the conditions outlined in section \ref{subsection: Bow data}, which is invariant under $\varphi$ up to a gauge transformation.

\begin{definition}\label{def:Invariance}
    An element $X\in\mathcal{BD}$ within the space of bow data is \textbf{symmetric} under a group action $\varphi:G\times\mathcal{BD}\rightarrow\mathcal{BD}$ if, for all $h\in G$ there exists a compensating gauge transformation $g\in\mathcal{G}$ such that $\varphi(h,g\cdot X)=X$.
\end{definition}

There are 3 group actions we wish to consider which leave the bow conditions invariant. Each corresponds to an action on the Taub-NUT instanton the bow data produces via the bow transform:

\paragraph{Rotations}

The Dirac monopole in the Taub-NUT metric \eqref{eq:Taub-NUT metric} vanishes along a line containing $\vec{q}$; by a choice of coordinates we can fix this line to be the first axis with $\vec{q}=(q,0,0)$, we then have the action on the bow data:
\begin{equation}\label{eq:Rotation action}
    \Theta:\begin{pmatrix}
        T^0 \\ \vec{T} \\ B_\pm \\ I_\pm \\ J_\pm
    \end{pmatrix}\mapsto\begin{pmatrix}
        T^0+\frac{i}{l}\theta\operatorname{id}\\ R_1(\theta)\vec{T} \\ B_\pm e^{i\frac{\theta}{2}} \\ I_\pm e^{i\frac{\theta}{2}} \\ J_\pm e^{i\frac{\theta}{2}}
    \end{pmatrix}
\end{equation}
where $R_1(\theta)\in SO(3)$ is a rotation around the first axis by an angle $\theta$. Through the bow transform this correspond to a rotation in the Taub-NUT instanton together with a translation in $\tau$: $\varphi:(\vec{t},\tau)\mapsto\left(R_1(\theta)\vec{t},\tau+\frac{1}{l}\theta\right)$.

\paragraph{Circle fibre translations}

\begin{equation}\label{eq:Circle fibre translation action}
    \mathcal{A}:\begin{pmatrix}
        T^0 \\ \vec{T} \\ B_\pm \\ I_\pm \\ J_\pm
    \end{pmatrix}\mapsto\begin{pmatrix}
        T^0+ia\operatorname{id} \\ \vec{T} \\ B_\pm \\ I_\pm \\ J_\pm
    \end{pmatrix}
\end{equation}
for $a\in\mathbb{R}$ corresponds with translations in the $\tau$ coordinate, $\tau\mapsto\tau+a$.

% \paragraph{Phase actions}

% \begin{equation}\label{eq:Unframed gauge transformation action}
%     \mathcal{B}:\begin{pmatrix}
%         T^0_{1,3} \\ T^0_2 \\ \vec{T} \\ B_\pm \\ I_\pm \\ J_\pm
%     \end{pmatrix}\mapsto\begin{pmatrix}
%         T_{1,3}^0 \\ T_2^0+ib\operatorname{id} \\ \vec{T} \\ B_\pm \\ I_\pm \\ J_\pm
%     \end{pmatrix}
% \end{equation}
% for $b\in\mathbb{R}$ corresponds with unframed gauge transformations, ie gauge transformations on the Taub-NUT instanton that do not act as the identity at infinity, but do preserve the splitting of the bundle $E$ into line bundles.

There is an additional action:
\begin{equation}\label{eq:Unframed gauge transformation action}
    \mathcal{B}:\begin{pmatrix}
        T^0_{1,3} \\ T^0_2 \\ \vec{T} \\ B_\pm \\ I_\pm \\ J_\pm
    \end{pmatrix}\mapsto\begin{pmatrix}
        T_{1,3}^0 \\ T_2^0+ib\operatorname{id} \\ \vec{T} \\ B_\pm \\ I_\pm \\ J_\pm
    \end{pmatrix}
\end{equation}
for $b\in\mathbb{R}$, which acts as a translation in $T^0$ only on the middle interval.

\subsubsection*{Symmetric bow data}\label{subsubsection: symmetric bow data}

By considering $a,b$ to be functions of $\theta$, we wish to find bow data $\mathcal{BD}$, satisfying the bow conditions, such that
\begin{equation}\label{eq:Symmetry equation}
    (\Theta\circ\mathcal{A}\circ\mathcal{B}\circ g)(\theta)\cdot\mathcal{BD}=\mathcal{BD}
\end{equation}
where $g$ is a `compensating gauge transformation'.
We can always choose to work in a gauge in which $T_P^0\equiv0$ across all intervals. We can consider the action $\varphi=\Theta\circ\mathcal{A}\circ\mathcal{B}$ in this gauge by defining
\begin{equation}
    \phi(\theta)=-\left(\frac{\theta}{2}+a(\theta)l+2b(\theta)\lambda\right)\qquad;\qquad \psi(\theta)=-\left(\frac{\lambda\theta}{2l}+a(\theta)\lambda+b(\theta)\lambda\right)
\end{equation}
then, the action of $\varphi$ on the bow data is
\begin{equation}
    \varphi:\begin{pmatrix}
         \vec{T} \\ B_\pm \\ I_\pm \\ J_\pm
    \end{pmatrix}\mapsto\begin{pmatrix}
         R_1(\theta)\vec{T} \\ B_\pm e^{i(\frac{\theta}{2}\pm\phi)} \\ I_\pm e^{i(\frac{\theta}{2}\pm\psi)} \\ J_\pm e^{i(\frac{\theta}{2}\mp\psi)}
    \end{pmatrix}.
\end{equation}
The actions of $\phi,\psi$ are then natural phase actions on the fundamental and bifundamental data, and we can therefore interpret the action \eqref{eq:Unframed gauge transformation action} as a combination of a global phase action (corresponding to the action of $\psi$ on the fundamental data) on the instanton together with a circle fibre translation. The global phase action acts via unframed gauge transformations on the instanton, these are gauge transformations that preserve the bundle splitting \eqref{eq:W decomposition} but do not act as the identity at spatial infinity.

In general, the compensating gauge transformation in definition \ref{def:Invariance} must be an $N$-dimensional unitary representation of the group, and since $N=2$, in our case must take the form
\begin{equation}\label{eq:Compensating gauge transformation}
    g=\exp\left(i\left(\frac{K\theta}{2}\sigma^2+\omega\right)\right)=\rho_{K}e^{i\omega}\qquad:\qquad \rho_{K}=\begin{pmatrix}
        \cos\frac{K\theta}{2} & \sin\frac{K\theta}{2} \\ -\sin\frac{K\theta}{2} & \cos\frac{K\theta}{2}
    \end{pmatrix}
\end{equation}
where $K\in\mathbb{Z}$ and $\omega\in\mathbb{R}$. We shall see \ref{Theorem: classification} that there are only 3 values of $K\in\mathbb{Z}$ that allow for solutions of the bow conditions: $K=0,\pm1$, the cases $K=\pm1$ are gauge equivalent, so we only need to consider one of them. This gives us 2 cases to consider: the `trivial compensating gauge' in which $K=0$, and the `non-trivial compensating gauge' where $K=1$, we will consider each of these in turn, and show that all $U(1)$-symmetric bow data of type $(2,2,2)$ can be classified by its precise symmetry.

\begin{remark}
   The action:
   \begin{equation}\label{eq:Reflections}
       F:\begin{pmatrix}
           T^0(s) \\ \vec{T}(s) \\ B_\pm \\ I_\pm \\ J_\pm
       \end{pmatrix}\mapsto\begin{pmatrix}
           -T^0(-s) \\ -\vec{T}(-s) \\ \mp B_\mp \\ I_\mp \\ J_\mp
       \end{pmatrix}
   \end{equation}
   also leaves the bow conditions invariant, and corresponds with a parity transformation $(\vec{t},\tau)\mapsto(-\vec{t},-\tau)$ in the Taub-NUT space. This is not an action under which we wish to find invariants, but will be used in the classification of additional solutions.
\end{remark}

We will see that valid solutions only exist if $\phi=\pm\frac{\theta}{2}$, with the $\pm$ cases related by the action \eqref{eq:Reflections}, as such, we will only consider $\phi=\frac{\theta}{2}$.

\subsection{Trivial compensating gauge transformation}\label{subsection: Trivial compensating gauge}

In the case $K=0$, the compensating gauge transformation reduces to
\begin{equation}\label{eq:Trivial compensating gauge transformation}
    g=e^{i\omega}\operatorname{id}
\end{equation}
and the symmetry equation becomes:
\begin{equation}\label{eq:Trivial compensating gauge action}
    (\varphi\circ g)\cdot\begin{pmatrix}
        \alpha_P \\ \beta_P \\ B_\pm \\ I_\pm \\ J_\pm
    \end{pmatrix}=\begin{pmatrix}
        \alpha_P \\ \beta_P e^{i\theta} \\ B_\pm e^{i(\frac{\theta}{2}\pm\frac{\theta}{2})} \\ I_\pm e^{i(\frac{\theta}{2}\pm\psi-\omega)} \\ J_\pm e^{i(\frac{\theta}{2}\mp\psi+\omega)}
    \end{pmatrix}=\begin{pmatrix}
        \alpha_P \\ \beta_P \\ B_\pm \\ I_\pm \\ J_\pm
    \end{pmatrix}
\end{equation}
for all $\theta\in[0,2\pi)$. We have solutions for $\psi=\pm\frac{\theta}{2}$, and $\omega=0$. In all cases $\beta_P=0$ for all $s\in\mathcal{I}_P$, and so Nahm's equations simplify to
\begin{equation}\label{eq:Simplified Nahm's equations for trivial compensating gauge}
    \frac{d}{ds}\alpha_P(s)=0
\end{equation}
so, $\alpha_P$ is a constant on each interval, which can be made Hermitian by a choice of gauge. We can also see that $B_+=0$. In the 2 cases for $\psi$, the bow data must take the form:
\begin{align*}
    \begin{pmatrix}
        \alpha_P, & 0, & 0, & B_-, & 0, & I_-, & J_+, & 0
    \end{pmatrix}\qquad&:\qquad \psi=\frac{\theta}{2}\\
    \begin{pmatrix}
        \alpha_P, & 0, & 0, & B_-, & I_+, & 0, &0, & J_-
    \end{pmatrix}\qquad&:\qquad \psi=-\frac{\theta}{2}.
\end{align*}

% The action $F$ \eqref{eq:Reflections} maps a set of bow data $\mathcal{BD}$ symmetric under the action $\varphi(\theta;\phi=\frac{\theta}{2};\psi=\frac{\theta}{2})$ to a set of bow data $F(\mathcal{BD})=\tilde{\mathcal{BD}}$ symmetric under the action $\varphi(\theta;\phi=-\frac{\theta}{2},\psi=-\frac{\theta}{2})$. Therefore, if we solve the bow conditions for one of these cases, we can construct the bow data for the other by an application of $F$. A similar argument relates the other 2 cases. We therefore only need to solve the bow conditions in 2 of the 4 cases: $\phi=\frac{\theta}{2}$, $\psi=\pm\frac{\theta}{2}$.

For $\psi=\frac{\theta}{2}$, the matching conditions become
\begin{equation}
    \begin{split}
        2\alpha_1=-B_-^\dagger B_-+2q\\
        2\alpha_1-2\alpha_2=I_-I_-^\dagger\\
        2\alpha_2-2\alpha_3=-J_+^\dagger J_+\\
        2\alpha_3=-B_-B_-^\dagger+2q
    \end{split}
\end{equation}
we can therefore fix $\alpha_P$, and we are left solving
\begin{equation}
    I_-I_-^\dagger-J_+^\dagger J_++[B_-^\dagger,B_-]=0.
\end{equation}
Similarly, when $\psi=-\frac{\theta}{2}$, we have
\begin{equation}
    \begin{split}
        2\alpha_1=-B_-^\dagger B_-+2q\\
        2\alpha_1-2\alpha_2=-J_-^\dagger J_-\\
        2\alpha_2-2\alpha_3=I_+I_+^\dagger\\
        2\alpha_3=-B_-B_-^\dagger+2q
    \end{split}
\end{equation}
and we are left solving
\begin{equation}
    I_+I_+^\dagger-J_-^\dagger J_-+[B_-^\dagger,B_-]=0.
\end{equation}
In both cases, solving the bow conditions amounts to solving the equation
\begin{equation}\label{eq:ADHM equation for trivial compensating gauge bow data}
    II^\dagger-J^\dagger J+[B^\dagger,B]=0
\end{equation}
so, we can solve both cases simultaneously. The reality condition implies, the matching data take the form
\begin{align}
    B=\begin{pmatrix}
        ib_0+b_3 & b_1-ib_2 \\ b_1+ib_2 & ib_0-b_3
    \end{pmatrix}e^{i\mu}=\begin{pmatrix}
        Z & W^* \\ W & -Z^*
    \end{pmatrix}e^{i\mu}\quad&:\quad b_\alpha\in\mathbb{R}_+,\quad \mu\in[0,2\pi),\quad Z,W\in\mathbb{C}\\
    I=\begin{pmatrix}
        \gamma_1 e^{i\chi_1} \\ \gamma_2 e^{i\chi_2}
    \end{pmatrix}\quad,\quad J=\begin{pmatrix}
        \gamma_1 e^{i\chi_1}&\gamma_2 e^{i\chi_2}
    \end{pmatrix}\quad&:\quad \gamma_j\in\mathbb{R}_+,\quad \chi_j\in[0,2\pi).
\end{align}
Considering \eqref{eq:ADHM equation for trivial compensating gauge bow data}, we see that
\begin{equation}
    II^\dagger-J^\dagger J+[B^\dagger,B]=\begin{pmatrix}
        0 & 2i\gamma_1\gamma_2\sin(\chi_1-\chi_2) \\ 2i\gamma_1\gamma_2\sin(\chi_2-\chi_1) & 0
    \end{pmatrix}=0;
\end{equation}
since $\gamma_1,\gamma_2>0$, we must have that $\chi_1=\chi_2=:\chi$.

The matching data therefore take the form
\begin{equation}
    I=\begin{pmatrix}
        \gamma_1 \\ \gamma_2
    \end{pmatrix}e^{i\chi} \quad,\quad J=\begin{pmatrix}
        \gamma_1 & \gamma_2
    \end{pmatrix}e^{i\chi}\quad,\quad B=\begin{pmatrix}
        Z & W^* \\ W & -Z^*
    \end{pmatrix}e^{i\mu}
\end{equation}
for $\gamma_2,\gamma_2\in\mathbb{R}_+$, $\mu,\chi\in[0,2\pi)$, and $(Z,W)\in\mathbb{C}^2\setminus\lbrace0\rbrace$. The Nahm data take the form
\begin{equation}
    \begin{split}
        2\alpha_1=2\alpha_3=\left(2q-|Z|^2-|W|^2\right)\operatorname{id}\\
        2\alpha_2=-\operatorname{diag}\left(\gamma_1^2,\gamma_2^2\right)+\left(2q-|Z|^2-|W|^2\right)\operatorname{id}
    \end{split}\qquad:\qquad \psi=\frac{\theta}{2}
\end{equation}
and
\begin{equation}
    \begin{split}
        2\alpha_1=2\alpha_3=\left(2q-|Z|^2-|W|^2\right)\operatorname{id}\\
        2\alpha_2=\operatorname{diag}\left(\gamma_1^2,\gamma_2^2\right)+\left(2q-|Z|^2-|W|^2\right)\operatorname{id}
    \end{split}\qquad:\qquad \psi=-\frac{\theta}{2}.
\end{equation}

\subsection{Non-trivial compensating gauge and global phase action $\psi\neq0$}\label{subsection: Non-trivial compensating gauge and psi non-zero}

We now consider $K=1$, in which the compensating gauge transformation is
\begin{equation}
    g=\exp\left(i\left(\frac{\theta}{2}\sigma^2+\omega\right)\right)=\rho e^{i\omega}
\end{equation}
then, the symmetry equations for the bow data are:
\begin{equation}
    (\varphi\circ g)\cdot\begin{pmatrix}
        \alpha_P \\ \beta_P \\ B_\pm \\ I_\pm  \\ J_\pm
    \end{pmatrix}=\begin{pmatrix}
        \rho^{-1}\alpha_P\rho \\ \rho^{-1}\beta_P\rho e^{i\theta} \\ \rho^{-1}B_\pm\rho e^{i(\frac{\theta}{2}\pm\frac{\theta}{2})} \\ \rho^{-1} I_\pm e^{i(\frac{\theta}{2}\pm\psi-\omega)} \\ J_\pm\rho e^{i(\frac{\theta}{2}\mp\psi+\omega)}
    \end{pmatrix}=\begin{pmatrix}
        \alpha_P \\ \beta_P \\ B_\pm \\ I_\pm \\ J_\pm
    \end{pmatrix}.
\end{equation}
There are valid solutions for $\psi=0,\pm\theta$, and $\omega=0$. For now we will focus on $\psi\neq0$.
% we will, for now assume that $\psi\neq0$, then the only values of $\phi,\psi,\omega$ that allow for valid solutions are: $\phi=\pm\frac{\theta}{2}$, $\psi=\pm\theta$, and $\omega=0$, there is also a pair of valid solutions for $\phi=\pm\frac{\theta}{2}$, $\psi=0$, $\omega=0$, but we will address these separately in the next section. By a similar argument as in the trivial compensating gauge case, we only need to consider $\phi=\frac{\theta}{2}$.

The solution to the symmetry equations and reality condition for the Nahm and bifundamental data, up to a choice of gauge, is
\begin{align}
    \alpha_1(s)=u_1(s)\operatorname{id}+v_1(s)\sigma^2\qquad&;\qquad\beta_1(s)=w_1(s)(\sigma^3+i\sigma^1)\label{eq:alpha_1,beta_1}\\
    \alpha_2(s)=u_2(s)\operatorname{id}+v_2(s)\sigma^2\qquad&;\qquad \beta_2(s)=w_2(s)(\sigma^3+i\sigma^1)\label{eq:alpha_2,beta_3}\\
    \alpha_3(s)=u_1(-s)\operatorname{id}-v_1(-s)\sigma^2\qquad&;\qquad \beta_3(s)=w_1(-s)(\sigma^3+i\sigma^1)\label{eq:alpha_3,beta_3}\\
    B_+=b_0\operatorname{id}\qquad&;\qquad B_-=b_1(\sigma^3+i\sigma^1)\label{eq:B_+,B_-}
\end{align}
for $u_P,v_P:\mathcal{I}_P\rightarrow\mathbb{R}$, $w_P:\mathcal{I}_P\rightarrow\mathbb{C}$, $b_0\in\mathbb{C}$, and $b_1\in\mathbb{R}$, we must also have that $u_2,w_2$ are even functions, and $v_2$ is odd by the reality condition. There is an alternative solution to the reality condition for the bifundamental data in which $B_+=b_0\operatorname{id}+b_2\sigma^2$, $B_-=0$, however this doesn't produce solutions. Nahm's equations for the symmetric data are:
\begin{equation}
    u_P^\prime=0\qquad,\qquad v_P^\prime+2w_Pw_P^*=0\qquad;\qquad w_P^\prime+2v_Pw_P=0
\end{equation}
so, $u_P$ must be constant on each interval, and we have the general solutions
\begin{align}
    v_1(s)=-\frac{1}{2}\kappa_1\tan(\kappa_1(s+s_1))\qquad&;\qquad w_1(s)=\frac{1}{2}\kappa_1\sec(\kappa_1(s+s_1))e^{i\chi_1}\\
    v_2(s)=-\frac{1}{2}\kappa_2\tan(\kappa_2 s)\qquad&;\qquad w_2(s)=\frac{1}{2}\kappa_2\sec(\kappa_2 s)e^{i\chi_2}
\end{align}
for $s_1,\kappa_1,\kappa_2\in\mathbb{R}$, which we constrain to ensure the Nahm data are continuous on each interval, therefore we impose that
\begin{equation}\label{eq:psi=theta parameter constraints}
    0<\kappa_1<\frac{\pi}{l-2\lambda}\quad,\quad 0<\kappa_2<\frac{\pi}{2\lambda}\quad,\quad \frac{l}{2}-\frac{\pi}{2\kappa_1}<s_1<\frac{l}{2}.
\end{equation}

The solutions for the fundamental data depend on the value of $\psi$; let $V_1=\begin{pmatrix}
    i \\ 1
\end{pmatrix}$, and $V_2=\begin{pmatrix}
    -i \\ 1
\end{pmatrix}$, then, for the 2 possible values of $\psi$, the fundamental data take the form
\begin{itemize}
    \item for $\psi=-\theta$, we have:
    $
        I_+=\gamma e^{i\nu_+}V_1,\quad I_-=0,\quad J_+=0,\quad J_-=\gamma e^{i\nu_-}V_2^\dagger
    $
    \item for $\psi=\theta$, we have:
    $
        I_+=0,\quad I_-=\gamma e^{i\nu_-}V_1,\quad J_+=\gamma e^{i\nu_+}V_2^\dagger,\quad J_-=0
    $
\end{itemize}
For $\psi=\pm\theta$, the matching conditions become:
\begin{align}
    2u_1=|b_0|^2-2|b_1|^2+q\label{eq:psi=theta matching conditions equation 1}\\
    \kappa_1\tan(\kappa_1(\frac{l}{2}-s_1))=2|b_1|^2\label{eq:psi=theta matching conditions equation 2}\\
    \kappa_1\sec(\kappa_1(\frac{l}{2}-s_1))e^{i\chi_1}=2b_0b_1\label{eq:psi=theta matching conditions equation 3}\\
    2u_1-2u_2=\pm\gamma^2\label{eq:psi=theta matching conditions equation 4}\\
    \kappa_1\tan(\kappa_1(s_1-\lambda))+\kappa_2\tan(\kappa_2\lambda)=\gamma^2\label{eq:psi=theta matching conditions equation 5}\\
    \kappa_1\sec(\kappa_1(s_1-\lambda))e^{i\chi_1}-\kappa_2\sec(\kappa_2\lambda)e^{i\chi_2}=0.\label{eq:psi=theta matching conditions equation 6}
\end{align}
Let $b_0=\sqrt{2}R\sin(a)e^{i\mu_0}$ and $b_1=R\cos(a)$ for $R>0$ and $a\in[0,\frac{\pi}{2}]$. We note that we can choose a gauge in which $\nu_-=\nu_+$. Then, \eqref{eq:psi=theta matching conditions equation 1} and \eqref{eq:psi=theta matching conditions equation 4} give
\begin{equation}
    u_1=q-R^2\cos(2a)\qquad,\qquad u_2=q\mp\gamma^2-R^2\cos(2a)
\end{equation}
and we can solve \eqref{eq:psi=theta matching conditions equation 5} for $R^2$:
\begin{equation}
    R^2\cos^2(a)=\frac{1}{2}\kappa_1\tan\left(\kappa_1\left(\frac{l}{2}-s_1\right)\right).
\end{equation}
\eqref{eq:psi=theta matching conditions equation 6} and \eqref{eq:psi=theta matching conditions equation 3} imply that $\chi_1=\chi_2=\mu_0=:\mu$. We can rearrange \eqref{eq:psi=theta matching conditions equation 3} to
\begin{equation}
    \cot(a)=\sqrt{2}\sin\left(\kappa_1\left(\frac{l}{2}-s_1\right)\right)
\end{equation}
which can be solved for $a\in\left[\arctan\left(\frac{\sqrt{2}}{2}\right),\frac{\pi}{2}\right)$. And we can solve \eqref{eq:psi=theta matching conditions equation 6},
\begin{equation}\label{eq:psi=theta final equation}
    \kappa_1\sec(\kappa_1(s_1-\lambda))-\kappa_2\sec(\kappa_2\lambda)=0
\end{equation}
for $\kappa_2$: define the function $F:\left(0,\frac{\pi}{2\lambda}\right)\rightarrow\mathbb{R}_+$ by $F:x\mapsto x\sec(\lambda x)$, which is bijective. The inequalities \eqref{eq:psi=theta parameter constraints} imply
\begin{equation}
    -\frac{\pi}{2}<\kappa_1(s_1-\lambda)<\frac{\pi}{2}
\end{equation}
so, $\kappa_1\sec(\kappa_1(s_1-\lambda))$ is in the range of $F$. Therefore, there exists a unique $\kappa_2\in\left(0,\frac{\pi}{2\lambda}\right)$ solving \eqref{eq:psi=theta final equation} for any $0<\lambda<\frac{l}{2}$, and any $\kappa_1,s_1$ constrained as above.

It remains to solve \eqref{eq:psi=theta matching conditions equation 5}, which can only be solved for $\gamma^2$ if the left-hand-side is positive. Considering \eqref{eq:psi=theta final equation}, we can rearrange \eqref{eq:psi=theta matching conditions equation 5} to
\begin{equation}\label{eq:psi=theta equation for gamma squared}
    \gamma^2=2\kappa_1\sec(\kappa_1(s_1-\lambda))\left(\cos\left(\frac{\kappa_1(s_1-\lambda)-\kappa_2\lambda}{2}\right)\sin\left(\frac{\kappa_1(s_1-\lambda)+\kappa_2\lambda}{2}\right)\right).
\end{equation}
From above, we know that $\sec(\kappa_1(s_1-\lambda))>0$, and that
\begin{equation}
    -\frac{\pi}{2}<\frac{\kappa_1(s_1-\lambda)-\kappa_2\lambda}{2}<\frac{\pi}{2}
\end{equation}
so $\cos\left(\frac{\kappa_1(s_1-\lambda)-\kappa_2\lambda}{2}\right)>0$. We also have that
\begin{equation}
    -\frac{\pi}{4}<\frac{\kappa_1(s_1-\lambda)+\kappa_2\lambda}{2}<\frac{\pi}{2}
\end{equation}
so, the sign of the right-hand-side of \eqref{eq:psi=theta equation for gamma squared} is entirely determined by the sign of $\kappa_1(s_1-\lambda)+\kappa_2\lambda$. Let $A=\kappa_1(s_1-\lambda)$ and $B=\kappa_2\lambda$, then

\begin{lemma}\label{A+B=0 iff s_1=0}
    Within the parameter space $A+B=0$ if and only if $s_1=0$; moreover, this gives $\kappa_1=\kappa_2$ and requires $0<\kappa_1<\frac{\pi}{l}$.
\end{lemma}

\begin{proof}
    If $A+B=0$, then $-\kappa_2\lambda=\kappa_1(s_1-\lambda)$, so by \eqref{eq:psi=theta final equation}
    \begin{equation*}
        \kappa_1\sec(\kappa_1(s_1-\lambda))-\kappa_2\sec(\kappa_2\lambda)=(\kappa_1-\kappa_2)\sec(\kappa_2\lambda)=0
    \end{equation*}
    therefore $\kappa_1=\kappa_2$, and $A+B=\kappa_1(s_1-\lambda)+\kappa_2\lambda=\kappa_1s_1=0$, therefore $s_1=0$.

    Conversely, if $s_1=0$, then $0<\kappa_1<\frac{\pi}{l}<\frac{\pi}{2\lambda}$, so \eqref{eq:psi=theta final equation} takes the form $F(\kappa_1)=F(\kappa_2)$ for $F:\left(0,\frac{\pi}{l}\right)\rightarrow\mathbb{R}_+$ defined by $F:x\mapsto x\sec(\lambda x)$, which is injective, so $\kappa_1=\kappa_2$. Therefore $A+B=\kappa_1(0-\lambda)+\kappa_1\lambda=0$.
\end{proof}

We note that $\kappa_2$ depends on $\kappa_1$ and $s_1$ from \eqref{eq:psi=theta final equation}, while $\kappa_1$ does not depend directly on $s_1$, but only in the definition of its domain. Differentiating \eqref{eq:psi=theta final equation} with respect to $s_1$, we have
\begin{equation}
    \kappa_1^2\tan(\kappa_1(s_1-\lambda))\sec(\kappa_1(s_1-\lambda))=\frac{d\kappa_2}{ds_1}\sec(\kappa_2\lambda)(1+\kappa_2\lambda\tan(\kappa_2\lambda))
\end{equation}
so
\begin{equation}
    \frac{d}{ds_1}(A+B)=\kappa_1\left(1-\frac{\kappa_1\lambda\tan(\kappa_1(s_1-\lambda))}{\sec(\kappa_2\lambda)(1+\kappa_2\lambda\tan(\kappa_2\lambda))}\right)
\end{equation}
in particular, by Lemma \ref{A+B=0 iff s_1=0}:
\begin{equation}
    \frac{d}{ds}\bigg|_{s_1=0}(A+B)=\kappa_1\left(1-\frac{\kappa_1\lambda\tan(\kappa_1\lambda)}{1+\kappa_1\lambda\tan(\kappa_1\lambda)}\right)>0
\end{equation}
So, $(A+B)(s_1)$ is an increasing function at $s_1=0$, and vanishes if and only if $s_1=0$ by Lemma \ref{A+B=0 iff s_1=0}. Since $A+B>0$ when $s_1=\lambda$, by continuity, $A+B>0$ for all $s_1>0$, and similarly $A+B<0$ for all $s_1<0$. Therefore \eqref{eq:psi=theta matching conditions equation 5} admits solutions if and only if $s_1>0$

\subsubsection*{Summary}\label{subsubsection: summary of the data}

 We therefore have a 4-parameter family of bow data, parameterised by $\mu_0,\nu_-\in[0,2\pi)$, $\kappa_1\in\left(0,\frac{\pi}{l-2\lambda}\right)$, and $s_1\in\left(\max\left(0,\frac{l}{2}-\frac{\pi}{2\kappa_1}\right),\frac{l}{2}\right)$, which explicitly takes the form:
\begin{align*}
        &\alpha_1(s)=\begin{pmatrix}
            q+\frac{\kappa_1}{2}\cot(2C) & i\frac{\kappa_1}{2}\tan(\kappa_1(s+s_1))\\
            -i\frac{\kappa_1}{2}\tan(\kappa_1(s+s_1)) & q+\frac{\kappa_1}{2}\cot(2C)
        \end{pmatrix}\\
        &\beta_1(s)=\frac{\kappa_1}{2}\sec(\kappa_1(s+s_1))e^{i\mu}\begin{pmatrix}
            1 & i \\ i & -1
        \end{pmatrix}\\
        &\alpha_2(s)=\begin{pmatrix}
            q+\frac{\kappa_1}{2}(\cot(2C)\mp\sec(A)(\sin(A)+\sin(B))) & i\frac{\kappa_2}{2}\tan(\kappa_2s)\\
            -i\frac{\kappa_2}{2}\tan(\kappa_2s) & q+\frac{\kappa_1}{2}(\cot(2C)\mp\sec(A)(\sin(A)+\sin(B)))
        \end{pmatrix}\\
        &\beta_2(s)=\frac{\kappa_2}{2}\sec(\kappa_2s)e^{i\mu}\begin{pmatrix}
            1 & i \\ i & -1
        \end{pmatrix}\\
        &\alpha_3(s)=\begin{pmatrix}
            q+\frac{\kappa_1}{2}\cot(2C) & i\frac{\kappa_1}{2}\tan(\kappa_1(s-s_1))\\
            -i\frac{\kappa_1}{2}\tan(\kappa_1(s-s_1)) & q+\frac{\kappa_1}{2}\cot(2C)
        \end{pmatrix}\\
        &\beta_3(s)=\frac{\kappa_1}{2}\sec(\kappa_1(s-s_1))e^{i\mu}\begin{pmatrix}
            1 & i \\ i & -1
        \end{pmatrix}\\
        & B_+=\sqrt{\kappa_1\csc(2C)}e^{i\mu}\operatorname{id},\quad B_-=\sqrt{\frac{\kappa_1\tan(C)}{2}}\begin{pmatrix}
            1 & i \\ i & -1
        \end{pmatrix}\\
        &I_\pm=0,\quad I_\mp=\sqrt{\kappa_1\sec(A)(\sin(A)+\cos(B))}e^{i\sigma}\begin{pmatrix}
            i \\ 1
        \end{pmatrix}\\
        &J_\pm=\sqrt{\kappa_1\sec(A)(\sin(A)+\cos(B))}e^{i\sigma}\begin{pmatrix}
            i & 1
        \end{pmatrix},\quad J_\mp=0
    \end{align*}
where $A=\kappa_1(s_1-\lambda)$, $B=\kappa_2\lambda$, $C=\kappa_1(\frac{l}{2}-s_1)$, and $\kappa_2$ solves
\begin{equation}
        \kappa_1\sec(\kappa_1(s_1-\lambda))-\kappa_2\sec(\kappa_2\lambda)=0.
\end{equation}

\subsection{Non-trivial compensating gauge and trivial global phase $\psi=0$}\label{subsection: non-trivial compensating gauge and psi=0}

Consider $K=1$ and $\psi=0$. The Nahm and bifundamental data take the same forms as in the $\psi=\pm\theta$ cases, \eqref{eq:alpha_1,beta_1}-\eqref{eq:B_+,B_-}. However, now the fundamental data take the form
\begin{equation}
    I_+=\gamma_+e^{i\sigma_+}\begin{pmatrix}
        -i \\ 1
    \end{pmatrix},\quad I_-=\gamma_-e^{i\sigma_-}\begin{pmatrix}
        -i \\ 1
    \end{pmatrix},\quad J_+=\gamma_-e^{i\sigma_-}\begin{pmatrix}
        1 & i
    \end{pmatrix},\quad J_-=-\gamma_+e^{i\sigma_+}\begin{pmatrix}
        1 & i
    \end{pmatrix}
\end{equation}
therefore the matching conditions become
\begin{align}
    &2u_1=b_0^2-2b_1^2+2q\label{eq:psi=0 equation 1}\\
    &\kappa_1\tan(\kappa_1(\frac{l}{2}-s_1))=2b_1^2\label{eq:psi=0 equation 2}\\
    &\kappa_1\sec(\kappa_1(\frac{l}{2}-s_1))e^{i\chi_1}=2b_0b_1e^{i\mu}\label{eq:psi=0 equation 3}\\
    &2u_1-2u_2=\gamma_-^2-\gamma_+^2\label{eq:psi=0 equation 4}\\
    &\kappa_1\tan(\kappa_1(s_1-\lambda))+\kappa_2\tan(\kappa_2\lambda)=-(\gamma_-^2+\gamma_+^2)\label{eq:psi=0 equation 5}\\
    &\kappa_1\sec(\kappa_1(s_1-\lambda))e^{i\chi_1}-\kappa_2\sec(\kappa_2\lambda)e^{i\chi_2}=-i\gamma_-\gamma_+e^{i(\sigma_-+\sigma_+)}.\label{eq:psi=0 equation 6}
\end{align}
Let
\begin{equation*}
    b_0=\sqrt{2}R\sin(a),\quad b_1=R\cos(a),\quad \gamma_-=\sqrt{2}r\sin(c),\quad \gamma_+=\sqrt{2}r\cos(c)
\end{equation*}
for $r,R\in\mathbb{R}_+$, and $a,c\in[0,\frac{\pi}{2}]$. Then \eqref{eq:psi=0 equation 1} and \eqref{eq:psi=0 equation 4} imply
\begin{equation*}
    u_1=q-R^2\cos(2a)\qquad,\qquad u_2=q-R^2\cos(2a)+r^2\cos(2c).
\end{equation*}
\eqref{eq:psi=0 equation 3} implies that $\chi_1=\mu$, so \eqref{eq:psi=0 equation 2} and \eqref{eq:psi=0 equation 3} become
\begin{align}
    &\kappa_1\tan(\kappa_1(\frac{l}{2}-s_1))=2R^2\cos^2(a)\label{eq:psi=0 equation 7}\\
    &\kappa_1\sec(\kappa_1(\frac{l}{2}-s_1))=\sqrt{2}R^2\sin(2a).\label{eq:psi=0 equation 8}
\end{align}
\eqref{eq:psi=0 equation 7} can be solved for $R^2$ if $s_1<\frac{l}{2}$, while \eqref{eq:psi=0 equation 8} can be rearranged to
\begin{equation}
    \cot(a)=\sqrt{2}\sin\left(\kappa_1\left(\frac{l}{2}-s_1\right)\right)
\end{equation}
which can be solved for $a\in\left[\arctan\left(\frac{\sqrt{2}}{2}\right),\frac{\pi}{2}\right)$. We are therefore left solving \eqref{eq:psi=0 equation 5} and \eqref{eq:psi=0 equation 6}, taking the forms
\begin{align}
    &\kappa_1\tan(\kappa_1(s_1-\lambda))+\kappa_2\tan(\kappa_2\lambda)=-2r^2\label{eq:psi=0 equation 10}\\
    &\kappa_1\sec(\kappa_1(s_1-\lambda))e^{i\mu}-\kappa_2\sec(\kappa_2\lambda)e^{i\chi_2}=-2ir^2\sin(2c)e^{i(\sigma_-+\sigma_+)}.\label{eq:psi=0 equation 11}
\end{align}
\eqref{eq:psi=0 equation 10} can only be solved if the left-hand-side is negative, therefore $s_1<\lambda$, and
\begin{equation}
    \kappa_1\tan(\kappa_1(\lambda-s_1))>\kappa_2\tan(\kappa_2\lambda).\label{eq:psi=0 equation 12}
\end{equation}
\eqref{eq:psi=0 equation 11} can be rearranged into the form
\begin{equation}
    l_1e^{i\phi_1}+l_2e^{i\phi_2}+l_3\sin(2c)=0\label{eq:psi=0 equation 13}
\end{equation}
where
\begin{align}
    l_1=\kappa_1\sec(\kappa_1(\lambda-s_1))\qquad&,\qquad \phi_1=\mu-\sigma_--\sigma_+-\frac{\pi}{2}\label{eq:psi=0 equation 14}\\
    l_2=\kappa_2\sec(\kappa_2\lambda)\qquad&,\qquad \phi_2=\chi_2-\sigma_--\sigma_++\frac{\pi}{2}\label{eq:psi=0 equation 15}\\
    l_3=\kappa_1\tan(\kappa_1(\lambda-s_1))-\kappa_2\tan(\kappa_2\lambda)\qquad&.\label{eq:psi=0 equation 16}
\end{align}
An equation of the form \eqref{eq:psi=0 equation 13} appears in \cite{Cork2026ACalorons}, and can be solved for $\phi_1,\phi_2$ (and therefore $\mu,\chi_1$) if and only if the triangle inequalities
\begin{equation}
    l_1\le l_2+l_3\sin(2c),\quad l_2\le l_3\sin(2c)+l_1,\quad l_3\sin(2c)\le l_1+l_2\label{eq:psi=0 equation 17}
\end{equation}
are satisfied, along with the continuity constraints on $\kappa_1,\kappa_2,s_1$. The triangle inequalities are satisfied on an open interval $c\in(c_-,c_+)\subset[0,\frac{\pi}{4}]$ if and only if
\begin{equation}\label{eq:Triangle inequality}
    l_3>|l_1-l_2|.
\end{equation}
Therefore, considering \eqref{eq:psi=theta parameter constraints} and the refinement that $s_1<\lambda$, we wish to find the subset $S\subset\mathbb{R}^3$ of $(\kappa_1,\kappa_2,s_1)$ for which the following inequalities all hold:
\begin{align}
    &0<\kappa_1<\frac{\pi}{l-2\lambda}\label{eq:psi=0 equation 18}\\
    &0<\kappa_2<\frac{\pi}{2\lambda}\label{eq:psi=0 equation 19}\\
    &0<\nu:=\lambda-s_1<\frac{\pi}{2\kappa_1}-\left(\frac{l}{2}-\lambda\right)\label{eq:psi=0 equation 20}\\
    &0<\kappa_1\tan(\kappa_1\nu)-\kappa_2\tan(\kappa_2\lambda)-\left|\kappa_1\sec(\kappa_1\nu)-\kappa_2\sec(\kappa_2\lambda)\right|.\label{eq:psi=0 equation 21}
\end{align}
Inequality \eqref{eq:psi=0 equation 21} implies \eqref{eq:psi=0 equation 13}, and is equivalent to
\begin{align}
    0<F_\lambda(\kappa_2)-F_\nu(\kappa_1)\label{eq:psi=0 equation 22}\\
    0<G_\nu(\kappa_1)-G_\lambda(\kappa_2)\label{eq:psi=0 equation 23}
\end{align}
where
\begin{equation}
    F_\alpha(x):=\frac{x}{\sec(\alpha x)+\tan(\alpha x)}\qquad,\qquad G_\alpha(x):=\frac{x}{\sec(\alpha x)-\tan(\alpha x)}.\label{eq:psi=0 equation 24}
\end{equation}
For all $\alpha,\beta>0$, and for all $x\in\left(0,\frac{\pi}{2\alpha}\right)$, we have that:
\begin{enumerate}[label=\textbf{P.\arabic*},ref=P.\arabic*]
    \item \label{P1} $F_\alpha(x),G_\alpha(x)>0$
    \item \label{P2} $F_\alpha(0)=F_\alpha\left(\frac{\pi}{2\alpha}\right)=G_\alpha(0)=0$
    \item \label{P3} $G_\alpha^\prime(x)>0$, in particular $G_\alpha(x)$ is an increasing surjective function $\left(0,\frac{\pi}{2\alpha}\right)\rightarrow(0,\infty)$
    \item \label{P4} $F_\alpha$ has at most one turning point at $x=\frac{\Lambda}{\alpha}$, where $\Lambda\approx0.739285$ solves $\Lambda=\cos\Lambda$, this turning point is a maximum
    \item \label{P5} $\partial_\alpha F_\alpha<0$ and $\partial_\alpha G_\alpha>0$
    \item \label{P6} $F_{\beta}\left(\frac{\alpha}{\beta}x\right)=\frac{\alpha}{\beta}F_\alpha(x)$ and $G_\beta\left(\frac{\alpha}{\beta}x\right)=\frac{\alpha}{\beta}G_\alpha(x)$.
\end{enumerate}
We note that \eqref{eq:psi=0 equation 20} refines \eqref{eq:psi=0 equation 18} to $0<\kappa_1<\frac{\pi}{l-2s_1}=\frac{\pi}{l-2\lambda+2\nu}$ and, together with \eqref{eq:psi=0 equation 19} implies that $\kappa_1,\kappa_2$ lie in some subset of $\left(0,\frac{\pi}{2\nu}\right)$ and $\left(0,\frac{\pi}{2\lambda}\right)$ respectively, so (\ref{P1}-\ref{P6}) hold.

Now, define the functions
\begin{align}
    z_\Omega(\kappa_1,\kappa_2):=F_\lambda(\kappa_2)-F_\nu(\kappa_1)\label{eq:psi=0 equation 25}\\
    z_\Gamma(\kappa_1,\kappa_2):=G_\nu(\kappa_1)-G_\lambda(\kappa_2)\label{eq:psi=0 equation 26}
\end{align}
which define smooth surfaces in $\mathbb{R}^3$, so \eqref{eq:psi=0 equation 22} and \eqref{eq:psi=0 equation 23} are satisfied when these surfaces are in the upper half space.

For fixed $\alpha$, since $G_\alpha:\left(0,\frac{\pi}{2\alpha}\right)\rightarrow(0,\infty)$ is bijective, for all $(\lambda,\nu,\kappa_1)$ there exists a unique $\kappa_+(\lambda,\nu,\kappa_1)\in\left(0,\frac{\pi}{2\lambda}\right)$ such that $z_{\Gamma}(\kappa_1,\kappa_+)=0$, and since $G_\alpha$ is increasing we see that $z_\Gamma(\kappa_1,\kappa_2)>0$ if and only if $\kappa_2<\kappa_+$, which satisfies \eqref{eq:psi=0 equation 23}.

It therefore remains to solve \eqref{eq:psi=0 equation 22}. $z_\Omega$ has a single saddle point $P$ when $(\kappa_1,\kappa_2)=\left(\frac{\Lambda}{\nu},\frac{\Lambda}{\lambda}\right)$, at which
\begin{equation}
    P=z_\Omega\left(\frac{\Lambda}{\nu},\frac{\Lambda}{\lambda}\right)=\left(\frac{1}{\lambda}-\frac{1}{\nu}\right)(1-\sin\Lambda)\label{eq:psi=0 equation 27}
\end{equation}
$(1-\sin\Lambda)>0$, and the saddle point slopes toward the negative $z$-axis in $\kappa_2$ and the positive $z$-axis in $\kappa_1$. It is therefore important to consider the relative locations of $\lambda$ and $\nu$.

\subsubsection*{$\nu>\lambda$}

Here, we can see that $P>0$ in \eqref{eq:psi=0 equation 27}, and since $z_\Omega$ always increases in the $\kappa_1$ direction, we need only to refine the $\kappa_2$ interval to solve \eqref{eq:psi=0 equation 23} for all $0<\kappa_1<\frac{\pi}{l-2\lambda+2\nu}<\frac{\pi}{2\nu}<\frac{\pi}{2\lambda}$. Note that
\begin{equation}
    z_{\Omega}(\kappa_1,\kappa_1)=F_\lambda(\kappa_1)-F_\nu(\kappa_1)>0\label{eq:psi=0 equation 28}
\end{equation}
since $F_\alpha(x)$ is decreasing in $\alpha$. On the other hand
\begin{equation}
    z_\Omega(\kappa_1,0)=z_\Omega(\kappa_1,\frac{\pi}{2\lambda})=-F_\nu(\kappa_1)<0.\label{eq:psi=0 equation 29}
\end{equation}
Therefore, by the intermediate value theorem, and since $F$ has at most a single turning point, for all $\kappa_1\in\left(0,\frac{\pi}{l-2\lambda+2\nu}\right)$ there exist unique $0<K_-(\kappa_1)<\kappa_1$ and $\kappa_1<K_+(\kappa_1)<\frac{\pi}{2\lambda}$ such that $z_\Omega(\kappa_1,K_\pm)=0$, and $z_\Omega(\kappa_1,\kappa_2)>0$ for all $K_-<\kappa_2<K_+$.

Combining this with the condition from inequality \eqref{eq:psi=0 equation 23}, this constrains \begin{equation*}K_-<\kappa_2<\min\left(K_+,\kappa_+\right).\end{equation*}

Consider the line $\kappa_2=\frac{\nu}{\lambda}\kappa_1$, then by \ref{P6}
\begin{align}
    &z_\Omega\left(\kappa_1,\frac{\nu}{\lambda}\kappa_1\right)=F_\lambda\left(\frac{\nu}{\lambda}\kappa_1\right)-F_\nu(\kappa_1)=\left(\frac{\nu}{\lambda}-1\right)F_\nu(\kappa_1)>0\\
    &z_\Gamma\left(\kappa_1,\frac{\nu}{\lambda}\kappa_1\right)=G_\nu(\kappa_1)-G_\lambda\left(\frac{\nu}{\lambda}\kappa_1\right)=\left(1-\frac{\nu}{\lambda}\right)G_\nu(\kappa_1)<0
\end{align}
therefore $\kappa_+(\kappa_1)<\frac{\nu}{\lambda}\kappa_1<K_+(\kappa_1)$ for all $\kappa_1$, so we can refine the constraint to
\begin{equation}
    K_-<\kappa_2<\kappa_+.
\end{equation}
Since $z_\Gamma(\kappa_1,\kappa_1)=G_\nu(\kappa_1)-G_\lambda(\kappa_1)>0$, we can also see that $K_-<\kappa_1<\kappa_+$ for all $\kappa_1$, so the above region is non-empty.

\subsubsection*{$\nu=\lambda$}

In this case the saddle point lies on the curve $\kappa_1=\kappa_2$, and $P=0$. In this case inequality \eqref{eq:psi=0 equation 23} is equivalent to $\kappa_1>\kappa_2$ since $G_\lambda$ is increasing. This condition allows for solutions of \eqref{eq:psi=0 equation 22} so long as $\kappa_1$ is large enough that it is in a region where $F_\lambda$ is strictly decreasing, therefore $\kappa_1>\frac{\Lambda}{\lambda}$. Because we also require that $\kappa_1<\frac{\pi}{l-2\lambda+2\nu}=\frac{\pi}{l}$, this places a constraint on $\lambda$:
\begin{equation}
    \lambda>\frac{l\Lambda}{\pi}.\label{eq:psi=0 equation 30}
\end{equation}
As long as this constraint on $\kappa_1$, and therefore on $\lambda$, is held, then there must exist exactly one value $0<\kappa_-<\frac{\Lambda}{\lambda}$, such that $F_\lambda(\kappa_-)=F_\lambda(\kappa_1)$, and therefore we have that $z_\Omega(\kappa_1,\kappa_2)>0$ for all $\kappa_-<\kappa_2<\kappa_1$.

\subsubsection*{$0<\nu<\lambda$}

Here the saddle point has $P<0$, so the $\kappa_1$ region will need to be constrained for \eqref{eq:psi=0 equation 22} to be solved. For any $\kappa_2$ we have that
\begin{equation}
    z_\Omega(0,\kappa_2)=z_\Omega\left(\frac{\pi}{2\nu},\kappa_2\right)=F_\lambda(\kappa_2)>0.\label{eq:psi-0 equation 31}
\end{equation}
Therefore, for all $\kappa_2\in\left(0,\frac{\pi}{2\lambda}\right)$ there exist unique $K_-(\kappa_2),K_+(\kappa_2)\in\left(0,\frac{\pi}{2\nu}\right)$ such that
\begin{equation}
    0<K_-<\frac{\Lambda}{\nu}<K_+<\frac{\pi}{2\nu}
\end{equation}
and such that
\begin{equation}
    z_\Omega\left(K_-,\kappa_2\right)=z_\Omega\left(K_+,\kappa_2\right)=0.
\end{equation}
Therefore, the surface define by $z_\Omega$ is in the upper half-space for $\kappa_1\in(0,K_-)\cup(K_+,\frac{\pi}{2\nu})$.

For a fixed $\alpha$, since $G_\alpha:\left(0,\frac{\pi}{2\alpha}\right)\rightarrow(0,\infty)$ is bijective, then for all $\lambda,\nu,\kappa_2$ there exists a unique $\kappa_-(\lambda,\nu,\kappa_2)\in\left(0,\frac{\pi}{2\nu}\right)$ such that $z_\Gamma(\kappa_-,\kappa_2)=0$, and since $G_\alpha$ is increasing $z_\Gamma(\kappa_1,\kappa_2)>0$ if and only if $\kappa_1>\kappa_-$.

Consider the line $\kappa_1=\kappa_2$, then
\begin{align}
    &z_\Omega(\kappa_2,\kappa_2)=F_\lambda(\kappa_2)-F_\nu(\kappa_2)<0\\
    &z_\Gamma(\kappa_2,\kappa_2)=G_\nu(\kappa_2)-G_\lambda(\kappa_2)<0
\end{align}
so, for all $\kappa_2$
\begin{equation}
    K_-(\kappa_2)<\kappa_2<\kappa_-(\kappa_2).
\end{equation}
So, we may discount the region $\kappa_1<K_-$, since $z_\Gamma<0$ in this region.

Now consider $\kappa_1=\frac{\lambda}{\nu}\kappa_2$, then
\begin{align}
    &z_\Omega\left(\frac{\lambda}{\nu}\kappa_2,\kappa_2\right)=F_\lambda(\kappa_2)-F_\nu\left(\frac{\lambda}{\nu}\kappa_2\right)=\left(1-\frac{\lambda}{\nu}\right)F_\lambda(\kappa_2)<0\\
    &z_\Gamma\left(\frac{\lambda}{\nu}\kappa_2,\kappa_2\right)=G_\nu\left(\frac{\lambda}{\nu}\kappa_2\right)-G_\lambda(\kappa_2)=\left(\frac{\lambda}{\nu}-1\right)G_\lambda(\kappa_2)>0
\end{align}
so, for all $\kappa_2$, we have
\begin{equation}
    \kappa_-(\kappa_2)<\frac{\lambda}{\nu}\kappa_2<K_+(\kappa_2).
\end{equation}
Since $\kappa_1>K_+$ is a stronger constraint than $\kappa_1>\kappa_-$, we only need to consider $\kappa_1>K_+$.

Therefore, it remains to check that $K_+(\kappa_2)\in\left(0,\frac{\pi}{2\nu+l-2\lambda}\right)\subset\left(0,\frac{\pi}{2\nu}\right)$, the allowed region for $\kappa_1$. Since $F_\alpha(x)$ has a single turning point, $K_+$ is minimised when $\kappa_2=\frac{\Lambda}{\lambda}$, at which
\begin{equation}
    F_\nu(K_+)=F_\lambda\left(\frac{\Lambda}{\lambda}\right)=\frac{1}{\lambda}(1-\sin\Lambda)
\end{equation}
since $K_+>\frac{\Lambda}{\nu}$, and $\frac{\Lambda}{\nu}$ is a local maximum of $F_\nu$, we have that $K_+<\frac{\pi}{2\nu+l-2\lambda}=\frac{\pi}{L}$ if and only if
\begin{equation}\label{eq:Final inequality}
    F_\nu\left(\frac{\pi}{L}\right)<\frac{1}{\lambda}\left(1-\sin\Lambda\right).
\end{equation}
If this holds, then there exists a subset $\kappa_1\in\left(K_+,\frac{\pi}{L}\right]\subset\left(0,\frac{\pi}{2\nu}\right)$ on which there can exist solutions. However: 
\begin{equation}
    z_\Omega\left(\frac{\pi}{L},0\right)=z_\Omega\left(\frac{\pi}{L},\frac{\pi}{2\lambda}\right)=-F_\nu\left(\frac{\pi}{L}\right)<0
\end{equation}
so, we must constrain the region on which $\kappa_2$ can take values. We note that
\begin{equation}
    z_\Omega\left(\frac{\pi}{L},\frac{\Lambda}{\lambda}\right)=\frac{1}{\lambda}(1-\sin\Lambda)-F_\nu\left(\frac{\pi}{L}\right)>0
\end{equation}
therefore, there exist unique $\kappa_2^-,\kappa_2^+\in\left(0,\frac{\pi}{2\lambda}\right)$ such that
\begin{equation}
    0<\kappa_2^-<\frac{\Lambda}{\lambda}<\kappa_2^+<\frac{\pi}{2\lambda}
\end{equation}
and such that
\begin{equation}
    z_\Omega\left(\frac{\pi}{L},\kappa_2^-\right)=z_\Omega\left(\frac{\pi}{L},\kappa_2^+\right)=0.
\end{equation}
Therefore, if \eqref{eq:Final inequality} holds, there exist subsets
\begin{equation}
    \left(K_+,\frac{\pi}{L}\right)\subset\left(0,\frac{\pi}{2\nu}\right)\qquad,\qquad \left(\kappa_2^-,\kappa_2^+\right)\subset\left(0,\frac{\pi}{2\lambda}\right)
\end{equation}
on which $\kappa_1,\kappa_2$ can take values for which the triangle inequality is satisfied.

\subsubsection*{Summary}\label{subsubsection: Summary of results for psi=0}

We therefore have that the inequalities \eqref{eq:psi=0 equation 18}-\eqref{eq:psi=0 equation 21} can be solved on some non-empty, open subset $S\subset\mathbb{R}^3$, where the subset is determined by the preceding analysis. We can therefore solve \eqref{eq:psi=0 equation 13} for $\phi_1,\phi_2$; doing so leaves us with a 6-parameter family of solutions parameterised by
\begin{equation}
    \left(\kappa_1,\kappa_2,s_1\right)\in S\quad,\quad \sigma_-,\sigma_+\in[0,2\pi)\quad,\quad c\in(c_-,c_+)\subset\left(0,\frac{\pi}{4}\right)
\end{equation}
such that
\begin{align*}
    &\alpha_1(s)=\begin{pmatrix}
        q+\frac{\kappa_1}{2}\cot(2C) & \frac{i}{2}\kappa_1\tan(\kappa_1(s+s_1)) \\ -\frac{i}{2}\kappa_1\tan(\kappa_1(s+s_1)) & q+\frac{\kappa_1}{2}\cot(2C)
    \end{pmatrix}\\
    & \beta_1(s)=\frac{1}{2}\kappa_1\sec(\kappa_1(s+s_1))e^{i\chi_1}\begin{pmatrix}
        1 & i \\ i & -1
    \end{pmatrix}\\
    &\alpha_2(s)=\begin{pmatrix}
        q+\frac{1}{2}(\kappa_1\tan(A)-\kappa_2\tan(B))\cos(2c) & \frac{i}{2}\kappa_2\tan(\kappa_2s) \\ -\frac{i}{2}\kappa_2\tan(\kappa_2s) & q+\frac{1}{2}(\kappa_1\tan(A)-\kappa_2\tan(B))\cos(2c)
    \end{pmatrix}\\
    &\beta_2(s)=\frac{1}{2}\kappa_2\sec(\kappa_2s)e^{i\chi_2}\begin{pmatrix}
        1 & i \\ i & -1
    \end{pmatrix}\\
    &\alpha_3(s)=\begin{pmatrix}
        q+\frac{\kappa_1}{2}\cot(2C) & \frac{i}{2}\kappa_1\tan(\kappa_1(s-s_1)) \\ -\frac{i}{2}\kappa_1\tan(\kappa_1(s-s_1)) & q+\frac{\kappa_1}{2}\cot(2C)
    \end{pmatrix}\\
    &\beta_3(s)=\frac{1}{2}\sec(\kappa_1(s-s_1))e^{i\chi_1}\begin{pmatrix}
        1 & i \\ i & -1
    \end{pmatrix}\\
    & B_+=\sqrt{\kappa_1\csc(2C)}e^{i\chi_1}\operatorname{id},\qquad B_-=\sqrt{\frac{1}{2}\kappa_1\tan(C)}\begin{pmatrix}
        1 & i \\ i & -1
    \end{pmatrix}\\
    &I_+=\sqrt{\kappa_1\tan(A)-\kappa_2\tan(B)}\cos(c)e^{i\sigma_-}\begin{pmatrix}
        -i \\ 1
    \end{pmatrix},\qquad I_-=\sqrt{\kappa_1\tan(A)-\kappa_2\tan(B)}\sin(c)e^{i\sigma_+}\begin{pmatrix}
        -i \\ 1
    \end{pmatrix}\\
    &J_+=\sqrt{\kappa_1\tan(A)-\kappa_2\tan(B)}\sin(c)e^{i\sigma_-}\begin{pmatrix}
        1 & i
    \end{pmatrix},\quad J_-=-\sqrt{\kappa_1\tan(A)-\kappa_2\tan(B)}\cos(c)e^{i\sigma_+}\begin{pmatrix}
        1 & i
    \end{pmatrix}
\end{align*}
where $A=\kappa_1(\lambda-s_1)$, $B=\kappa_2\lambda$, $C=\kappa_1(\frac{l}{2}-s_1)$, and
\begin{equation*}
    \chi_1=\phi_1+\sigma_-\sigma_++\frac{\pi}{2}\qquad;\qquad \chi_2=\phi_2+\sigma_-+\sigma_+-\frac{\pi}{2}.
\end{equation*}

\section{U(1) invariant solutions of type (1,2,1)}\label{section: U(1) invariant solutions (1,2,1)}

For a type-$(1,2,1)$ bow, we wish to consider rotations \eqref{eq:Rotation action} and circle fibre translations \eqref{eq:Circle fibre translation action}, which act on the Nahm and bifundamental data as in the type-$(2,2,2)$ case, but leave $X$ invariant. The action
\begin{equation}
    \Phi:\begin{pmatrix}
        T^0 \\ \vec{T} \\ B_\pm \\ X
    \end{pmatrix}\mapsto\begin{pmatrix}
        T^0 \\ \vec{T} \\ B_\pm \\ Xe^{i\phi}
    \end{pmatrix}
\end{equation}
leaves the bow conditions invariant, however, it can always be compensated for by a $U(1)$ gauge transformation that leaves the rest of the bow data invariant, we will therefore not consider this action. Therefore, we wish to find bow data symmetric under the action
\begin{equation}
    \varphi:\begin{pmatrix}
        \vec{T} \\ B_\pm \\ X
    \end{pmatrix}\mapsto\begin{pmatrix}
        R_1(\theta)\vec{T} \\ B_\pm e^{i(\frac{\theta}{2}\pm\phi)} \\ X
    \end{pmatrix}
\end{equation}
with the compensating gauge transformation
\begin{equation}
    g=\begin{cases}
        g_2=\rho_K(\theta)e^{i\omega_2} \qquad&:\qquad s\in\mathcal{I}_2\\
        g_1=e^{i\omega_1}\qquad&:\qquad s\in\mathcal{I}_1,\mathcal{I}_3
    \end{cases}.
\end{equation}
Similar to the $(2,2,2)$ case, we need only to consider $K=0,1$, and $\phi=\frac{\theta}{2}$.
% Similar to the $(2,2,2)$ case, there are only 3 values of $K\in\mathbb{Z}$ that allow for solutions of the bow conditions: $K=0,\pm1$, the cases $K=\pm1$ are gauge equivalent, so we only need to consider one of them. So, we have 2 cases to consider: the `trivial compensating gauge' in which $K=0$, and the `non-trivial compensating gauge' where $K=1$.

% Similar to in the $(2,2,2)$-type bow we will later see that the only values of $\phi$ that allow for valid solutions are $\phi=\pm\frac{\theta}{2}$, with the $\pm$ cases related by the action of \eqref{eq:Reflections}, so will will only consider the positive case.

\subsection{Trivial compensating gauge transformation}\label{subsection: trivial compensating gauge (1,2,1)}

In the case $K=0$, the compensating gauge transformation reduces to
\begin{equation}
    g=\begin{cases}
        g_2=e^{i\omega_2}\operatorname{id}\qquad&:\qquad s\in\mathcal{I}_2\\
        g_1=e^{i\omega_1}\qquad&:\qquad s\in\mathcal{I}_1,\mathcal{I}_3
    \end{cases}
\end{equation}
and the symmetry equation becomes:
\begin{equation}
    \left(\varphi\circ g\right)\cdot\begin{pmatrix}
        \alpha_P \\ \beta_P \\ B_\pm \\ X
    \end{pmatrix}=\begin{pmatrix}
        \alpha_P \\ \beta_Pe^{i\theta} \\ B_\pm e^{i(\frac{\theta}{2}\pm\frac{\theta}{2})} \\ Xe^{i(\omega_1-\omega_2)}
    \end{pmatrix}=\begin{pmatrix}
        \alpha_P \\ \beta_P \\ B_\pm \\ X
    \end{pmatrix}
\end{equation}
for all $\theta\in[0,2\pi)$. We have solutions for $\omega_2=\omega_1$. In all cases $\beta_P=0$ for all $s\in\mathcal{I}_P$ for all intervals, and Nahm's equations simplify to
\begin{equation}
    \frac{d}{ds}\alpha_P(s)=0
\end{equation}
so, $\alpha_P$ is constant on each interval, which can be made Hermitian by a gauge choice. On the outer intervals $\mathcal{I}_{1,3}$ $\alpha_P$ is therefore simply a real number, and on $\mathcal{I}_2$ it is a constant Hermitian $2\times2$ matrix. And we must have that $B_+=0$

The reality condition implies that $\alpha_2$ must be symmetric, and therefore:
\begin{equation}
    \alpha_2=\begin{pmatrix}
        a & b \\ b & d
    \end{pmatrix}\qquad:\qquad a,b,d\in\mathbb{R}
\end{equation}
and therefore the matching conditions become:
\begin{equation}
    X^\dagger\alpha_2 X=\alpha_1\qquad,\qquad 2\alpha_1=2q-|B_-|^2.
\end{equation}
By a gauge choice we can fix $X=\begin{pmatrix}
    1 & 0
\end{pmatrix}^t$. Therefore, the fundamental matching condition implies:
\begin{equation}
    a=\alpha_1=q-\frac{1}{2}|B_-|^2.
\end{equation}
We therefore have a 4-parameter family of solutions parameterised by $a,b,d\in\mathbb{R}$, and $\mu:=\arg(B_-)\in[0,2\pi)$ such that
\begin{align*}
    \alpha_1=\alpha_3=a\qquad,\qquad \beta_1=\beta_3=0\qquad&,\qquad
    \alpha_2=\begin{pmatrix}
        a & b \\ b & d
    \end{pmatrix}\qquad,\qquad \beta_2=0\\
    B_+=0\qquad,\qquad B_-=\sqrt{2(q-a)}&e^{i\mu}\qquad,\qquad
    X=\begin{pmatrix}
        1 \\ 0
    \end{pmatrix}.
\end{align*}

\subsection{Non-trivial compensating gauge transformation}\label{subsection: non-trivial compensating gauge (1,2,1)}

For $K=1$, the compensating gauge transformation is
\begin{equation}
    g=\begin{cases}
        g_2=\rho e^{i\omega_2} \qquad&:\qquad s\in\mathcal{I}_2\\
        g_1=e^{i\omega_1}\qquad&:\qquad s\in\mathcal{I}_1,\mathcal{I}_3
    \end{cases}
\end{equation}
and the symmetry equations become:
\begin{equation}
    \left(\varphi\circ g\right)\cdot\begin{pmatrix}
        \alpha_{1,3} \\ \beta_{1,3} \\ \alpha_2 \\ \beta_2 \\ B_\pm \\ X
    \end{pmatrix}=\begin{pmatrix}
        \alpha_{1,3} \\ \beta_{1,3}e^{i\theta} \\ \rho^{-1}\alpha_2\rho \\ \rho^{-1}\beta_2\rho e^{i\theta} \\ B_\pm e^{i(\frac{\theta}{2}\pm\frac{\theta}{2})} \\ \rho^{-1}X e^{i(\omega_1-\omega_2)}
    \end{pmatrix}=\begin{pmatrix}
        \alpha_{1,3} \\ \beta_{1,3} \\ \alpha_2 \\ \beta_2 \\ B_\pm \\ X
    \end{pmatrix}.
\end{equation}
Which can be solved when $\omega_1-\omega_2=\pm\frac{\theta}{2}$. Therefore, up to a gauge choice the symmetry equations and the reality conditions are solved by
\begin{align}
    \alpha_{1}=\alpha_3\in\mathbb{R}\qquad&,\qquad \beta_1=\beta_3=0\\ \alpha_2(s)=u_2(s)\operatorname{id}+v_2(s)\sigma^2\qquad&,\qquad \beta_2(s)=w_2(s)(\sigma_3+i\sigma_1)\\
\end{align}
and
\begin{equation}
    B_+=0\qquad;\qquad X=\frac{1}{\sqrt{2}}\begin{pmatrix}
        \pm i \\ 1
    \end{pmatrix}\qquad:\qquad \omega_1-\omega_2=\mp\frac{\theta}{2}.
\end{equation}
Nahm's equations imply that $\alpha_1,\alpha_3$ are constants, and that $u_2$ is a constant and $v_2,w_2$ take the forms
\begin{equation}
    v_2(s)=-\frac{1}{2}\kappa_2\tan(\kappa_2s)\qquad;\qquad w_2(s)=\frac{1}{2}\kappa_2\sec(\kappa_2s)e^{i\chi_2}.
\end{equation}
The matching conditions become:
\begin{align}
    &X^\dagger\alpha_2(-\lambda)X=\alpha_1\quad\Rightarrow\quad u_2\pm\frac{1}{2}\kappa_2\tan(\kappa_2\lambda)=\alpha_1\\
    &\alpha_1=q-\frac{1}{2}|B_-|^2.
\end{align}
So, $u_2=-\frac{1}{2}\left(|B_-|^2\pm\kappa_2\tan(\kappa_2\lambda)\right)$, and therefore we have a 4-parameter family of solutions parameterised by $\kappa_2\in(0,\frac{\pi}{4\lambda})$, $\chi_2\in[0,2\pi)$, and $B_-\in\mathbb{C}\setminus\lbrace0\rbrace$. Then:
\begin{align*}
    &\alpha_1=\alpha_3=q-\frac{1}{2}|B_-|^2\qquad,\qquad\beta_1=\beta_3=0\\
    &\alpha_2(s)=q-\frac{1}{2}\left(|B_-|^2\pm\kappa_2\tan(\kappa_2\lambda)\right)\operatorname{id}-\frac{1}{2}\kappa_2\tan(\kappa_2s)\sigma^2\\
    &\beta_2(s)=\frac{1}{2}\kappa_2\sec(\kappa_2s)e^{i\chi_2}(\sigma^3+i\sigma^1)\\
    &B_+=0\qquad,\qquad B_-=B_-\\
    &X=\frac{1}{\sqrt{2}}\begin{pmatrix}
        \pm i\\1
    \end{pmatrix}.
\end{align*}

\section{U(1) invariant solutions of type (2,1,2)}\label{section: U(1) invariant solutions of type (2,1,2)}

For a type-$(2,1,2)$ bow, we again consider rotations \eqref{eq:Rotation action}, and circle fibre translations \eqref{eq:Circle fibre translation action}, but similar to the $(1,2,1)$ case there is no phase action acting only on the fundamental data. Therefore we wish to find bow data symmetric under
\begin{equation}
    \varphi:\begin{pmatrix}
        \alpha_P \\ \beta_P \\ B_\pm \\ X
    \end{pmatrix}\mapsto\begin{pmatrix}
        \alpha_P \\ \beta_P e^{i\theta} \\ B_\pm e^{i(\frac{\theta}{2}\pm\phi)} \\ X
    \end{pmatrix}
\end{equation}
with the compensating gauge transformation
\begin{equation}
    g=\begin{cases}
        g_2=e^{i\omega_2}\qquad&:\qquad s\in\mathcal{I}_2\\
        g_1=\rho_Ke^{i\omega_1}\qquad&:\qquad s\in\mathcal{I}_1,\mathcal{I}_3
    \end{cases}.
\end{equation}
As before we need to consider $K=0,1$, and only need to consider $\phi=\frac{\theta}{2}$.

\subsection{Trivial compensating gauge transformation}\label{subsection: (2,1,2) trivial}

When $K=0$, the compensating gauge transformation reduces to
\begin{equation}
    g=\begin{cases}
        g_2=e^{i\omega_2}\qquad&:\qquad s\in\mathcal{I}_2\\
        g_1=e^{i\omega_1}\operatorname{id}\qquad&:\qquad s\in\mathcal{I}_1,\mathcal{I}_3
    \end{cases}
\end{equation}
and the symmetry equations become:
\begin{equation}
    (\varphi\circ g)\cdot\begin{pmatrix}
        \alpha_P \\ \beta_P \\ B_\pm \\ X
    \end{pmatrix}\mapsto\begin{pmatrix}
        \alpha_P \\ \beta_P e^{i\theta} \\ B_\pm e^{i(\frac{\theta}{2}\pm\frac{\theta}{2})} \\ X e^{i(\omega_1-\omega_2)}
    \end{pmatrix}=\begin{pmatrix}
        \alpha_P \\ \beta_P \\ B_\pm \\ X
    \end{pmatrix}
\end{equation}
for all $\theta\in[0,2\pi)$. We have solutions for $\omega_2=\omega_1$. We see that $\beta_P=0$ for all $s\in\mathcal{I}_P$ for all intervals, $B_+=0$, and Nahm's equations simplify to
\begin{equation}
    \frac{d}{ds}\alpha_P(s)=0
\end{equation}
so, by a choice of gauge $\alpha_P$ are constant and hermitian. By a gauge choice
\begin{equation}
    B_-=\begin{pmatrix}
        Z & w \\ w & -Z^*
    \end{pmatrix}e^{i\mu}
\end{equation}
for $Z\in\mathbb{C}$, $w\in\mathbb{R}$, $\mu\in[0,2\pi)$, and $X=\begin{pmatrix}
    1 & 0
\end{pmatrix}^t$. The matching conditions become:
\begin{align}
    2\alpha_1=2q-B_-^\dagger B_-\quad&\Rightarrow\quad \alpha_1=\frac{1}{2}\left(2q-|Z|^2-w^2\right)\operatorname{id}_2\\
    X^\dagger \alpha_1X=\alpha_1\quad&\Rightarrow\quad \alpha_2=\frac{1}{2}\left(2q-|Z|^2-w^2\right).
\end{align}
We therefore have a 4-parameter family of solutions parameterised by $Z\in\mathbb{C}$, $w\in\mathbb{R}$, $\mu\in[0,2\pi)$ such that
\begin{align}
    \alpha_1=\alpha_3=\frac{1}{2}\left(2q-|Z|^2-w^2\right)\operatorname{id}\quad&,\quad \beta_1=\beta_3=0\quad,\quad \alpha_2=\frac{1}{2}\left(2q-|Z|^2-w^2\right)\quad,\quad \beta_2=0\\
    B_+=0\quad&,\quad B_-=\begin{pmatrix}
        Z & w \\ w & -Z^*
    \end{pmatrix}e^{i\mu}\quad,\quad X=\begin{pmatrix}
        1 \\ 0
    \end{pmatrix}
\end{align}

\subsection{Non-trivial compensating gauge transformation}\label{subsection: (2,1,2) non-trivial}

For $K=1$, the compensating gauge transformation is
\begin{equation}
    g=\begin{cases}
        g_2=e^{i\omega_2}\qquad&:\qquad s\in\mathcal{I}_2\\
        g_1=\rho e^{i\omega_1}\qquad&:\qquad s\in\mathcal{I}_1,\mathcal{I}_3
    \end{cases}
\end{equation}
and the symmetry equations become:
\begin{equation}
    (\varphi\circ g)\cdot\begin{pmatrix}
        \alpha_{1,3} \\ \beta_{1,3} \\ \alpha_2 \\ \beta_2 \\ B_\pm \\ X
    \end{pmatrix}=\begin{pmatrix}
        \rho^{-1}\alpha_{1,3}\rho \\ \rho^{-1}\beta_{1,3}\rho e^{i\theta} \\ \alpha_2 \\ \beta_2 e^{i\theta} \\ \rho^{-1}B_\pm\rho e^{i(\frac{\theta}{2}\pm\frac{\theta}{2})} \\ \rho^{-1} X e^{i(\omega_1-\omega_2)}
    \end{pmatrix}=\begin{pmatrix}
        \alpha_{1,3} \\ \beta_{1,3} \\ \alpha_2 \\ \beta_2 \\ B_\pm \\ X
    \end{pmatrix}.
\end{equation}
Which can be solved when $\omega_1-\omega_2=\pm\frac{\theta}{2}$. Therefore, up to a gauge choice the symmetry equations and the reality condition are solved by
\begin{align}
    \alpha_1(s)=u_1(s)\operatorname{id}+v_1(s)\sigma^2\qquad&,\qquad \beta_1(s)=w_1(s)(\sigma^3+i\sigma^1) \\ \alpha_3(s)=u_1(-s)\operatorname{id}-v_1(-s)\sigma^2\qquad&,\qquad \beta_3(s)=w_1(-s)(\sigma^3+i\sigma^1) 
\end{align}
with $\alpha_2\in\mathbb{R}$, $\beta_2=0$, and
\begin{equation}
    B_+=b_1(\sigma^3+i\sigma^1) \quad,\quad B_-=b_0\operatorname{id}\quad,\quad X=\frac{1}{\sqrt{2}}\begin{pmatrix}
        \pm i \\ 1
    \end{pmatrix}
\end{equation}
for $u_1,v_1:\mathcal{I}_1\rightarrow\mathbb{R}$, $w_1:\mathcal{I}_1\rightarrow\mathbb{C}$ continuous, and $b_0,b_1\in\mathbb{C}$. Then Nahm's equations imply that $u_1$ is constant, and
\begin{equation}
    v_1(s)=-\frac{1}{2}\kappa_1\tan(\kappa_1(s+s_1))\qquad,\qquad w_1(s)=\frac{1}{2}\sec(\kappa_1(s+s_1))e^{i\chi_1}.
\end{equation}
The matching conditions become:
\begin{align}
    &2u_1=|b_0|^2-2|b_1|^2+2q\\
    &\kappa_1\tan(\kappa_1(\frac{l}{2}-s_1))=2|b_1|^2\\
    &\kappa_1\sec(\kappa_1(\frac{l}{2}-s_1))e^{i\chi_1}=2b_0b_1\\
    &\alpha_2=2\left(u_1\pm\kappa_1\tan(\kappa_1(s_1-\lambda))\right).
\end{align}
Let $b_0=\sqrt{2}R\sin(a)e^{i\mu_0}$, $b_1=R\cos(a)e^{i\mu_1}$. Then $\chi_1=\mu_0+\mu_1$, and
\begin{equation}
    2u_q=-R^2\cos(2a)
\end{equation}
and
\begin{equation}
    R^2\cos^2(a)=\frac{1}{2}\kappa_1\tan(\kappa_1(\frac{l}{2}-s_1))\qquad,\qquad \cot(a)=\sqrt{2}\sin(\kappa_1(\frac{l}{2}-s_1))
\end{equation}
can be solved for $R,a$ if $s_1<\frac{l}{2}$. This gives a 4-parameter family of solutions parameterised by $\kappa_1\in\left(0,\frac{\pi}{l-2\lambda}\right)$, $s_1\in\left(\frac{l}{2}-\frac{\pi}{2\kappa_1},\frac{l}{2}\right)$, $\mu_0,\mu_1\in[0,2\pi)$.

% \subsection{Classification of solutions}

% \begin{theorem}
%     All $U(1)$ symmetric Taub-NUT instantons of type $(1,2,1)$ are classified into 6 families by their precise symmetry, determined by $\phi,\Delta\omega:=\omega_1-\omega_2$:
%     \begin{itemize}
%         \item 4-parameter families with $\phi=\pm\frac{\theta}{2}$, $\Delta\omega=0$, $K=0$
%         \item 4-parameter families with $\psi=\pm\frac{\theta}{2}$, $\Delta\omega=\pm\frac{\theta}{2}$, $K=1$
%     \end{itemize}
% \end{theorem}

\section{Classification of solutions}\label{section: classification of solutions}

\begin{theorem}\label{Theorem: classification}
    All $U(1)$-symmetric Taub-NUT instantons with $(k,m)=(2,0),(1,1),(2,-1)$ can be classified into 18 families by their precise symmetry, determined by $\phi$ and $\psi$. There are 10 families of $(2,2,2)$ instantons:
    \begin{itemize}
        \item 4 $8$-parameter families with $\phi=\pm\frac{\theta}{2}$, $\psi=\pm\frac{\theta}{2}$ and trivial compensating gauge
        \item 4  $4$-parameter families with $\phi=\pm\frac{\theta}{2}$, $\psi=\pm\theta$ and non-trivial compensating gauge
        \item 2  $6$-parameter families with $\phi=\pm\frac{\theta}{2}$, $\psi=0$ and non-trivial compensating gauge
    \end{itemize}
    4 families of $(1,2,1)$ instantons:
    \begin{itemize}
        \item 2 $4$-parameter families with $\phi=\pm\frac{\theta}{2}$ and trivial compensating gauge
        \item 2 $4$-parameter families with $\phi=\pm\frac{\theta}{2}$ and non-trivial compensating gauge
    \end{itemize}
    and 4 families of $(2,1,2)$ instantons:
    \begin{itemize}
        \item 2 $4$-parameter families with $\phi=\pm\frac{\theta}{2}$, and trivial compensating gauge
        \item 2 $4$-parameter families with $\phi=\pm\frac{\theta}{2}$, and non-trivial compensating gauge.
    \end{itemize}
\end{theorem}

\begin{proof}
    By a similar argument as in \cite{Harland2007LargeCalorons} it is possible to show that the only values of $K$ in the compensating gauge transformation \eqref{eq:Compensating gauge transformation} that allow for valid solutions of the symmetry equations are $K=0,\pm1$. This is done via representation theory by showing that these have non-zero dimension of the trivial subrepresentations of the $U(1)$ action on the Nahm data. The cases $K=\pm1$ are gauge equivalent, and therefore we only need to consider $K=0,1$. This argument works for all three bow types considered.

    We now consider the bifundamental symmetry equations:
    \begin{equation}
        (\varphi\circ g)\cdot\begin{pmatrix}
            B_- \\ B_+
        \end{pmatrix}=\begin{pmatrix}
            g^{-1} B_- g e^{i(\frac{\theta}{2}-\phi)} \\ g^{-1} B_+g e^{i(\frac{\theta}{2}+\phi)}
        \end{pmatrix}=\begin{pmatrix}
            B_- \\ B_+
        \end{pmatrix}
    \end{equation}
    We note that, for $(1,2,1)$ bows, since the bifundamental data are rank-1 the compensating gauge acts trivially on the data, similarly for $(2,2,2)$ and $(2,1,2)$ bows with $K=0$, the compensating gauge transformation commutes with the bifundamental data and so acts trivially. In these cases therefore we have the symmetry equations
    \begin{equation}
        \begin{pmatrix}
            B_- e^{i(\frac{\theta}{2}-\phi)}\\ B_+ e^{i(\frac{\theta}{2}+\phi)}
        \end{pmatrix}=\begin{pmatrix}
            B_- \\ B_+
        \end{pmatrix}.
    \end{equation}
    Since, to have valid bow data, $(B_-,B_+)\neq 0$, we can see that we must have $\phi=\pm\frac{\theta}{2}$. Now, in the $(2,2,2)$ and $(2,1,2)$ bows with $K=1$, we have
    \begin{equation}
        \begin{pmatrix}
            \rho^{-1}B_-\rho e^{i(\frac{\theta}{2}-\phi)} \\ \rho^{-1}B_+ \rho e^{i(\frac{\theta}{2}+\phi)}
        \end{pmatrix}=\begin{pmatrix}
            B_- \\ B_+
        \end{pmatrix}
    \end{equation}
    the equation
    \begin{equation}
        \rho^{-1}(\theta) A\rho(\theta)=Ae^{i\chi(\theta)}\qquad:\qquad \rho(\theta)=\begin{pmatrix}
            \cos\frac{\theta}{2} & \sin\frac{\theta}{2} \\ -\sin\frac{\theta}{2} & \cos\frac{\theta}{2}
        \end{pmatrix},
    \end{equation}
    for non-zero $A\in\mathbb{C}^{2\times 2}$, is only solvable when $\chi(\theta)=0,\pm\theta$, therefore, for the symmetry equations on both $B_\pm$ to be solvable we must have that $\phi=\pm\frac{\theta}{2}$.
    
    Note that, since we only require one of the $B_\pm$ to be non-zero, we could have $\phi=\pm\theta$ with $B_\pm=0$ and $B_\mp\neq0$. This, however, produces no valid solutions to the bow conditions: if either of the $B_\pm$ vanish, then the bifundamental matching condition \eqref{eq:Bifundamental matching condition -}, along with the symmetric form for $\beta_1$ \eqref{eq:alpha_1,beta_1} implies:
    \begin{equation}
        \beta_1(-\frac{l}{2})=\frac{1}{2}\kappa_1\sec(\kappa_1(\frac{l}{2}-s_1))e^{i\chi_1}(\sigma^3+i\sigma^1)=B_+B_-=0
    \end{equation}
    this would force $\kappa_1=0$, which in turn forces all the Nahm data to vanish except $\alpha_1=u_1\operatorname{id}_{\mathcal{R}_1}$ and $\alpha_2=u_2\operatorname{id}_{\mathcal{R}_2}$ constant. In this case the fundamental data must vanish, and therefore there are no valid solutions.

    We now consider the symmetry equations on the fundamental data. Since there is no phase action on the fundamental data (as shown in section \ref{section: U(1) invariant solutions (1,2,1)}) we only need to consider $(2,2,2)$ type bow data. The symmetry equations are
    % We now consider the symmetry equations on the fundamental data. Since $\varphi$ doesn't act on the fundamental data in the $(1,2,1)$ and $(2,1,2)$ cases, by the argument in \ref{section: U(1) invariant solutions (1,2,1)}, we only need to consider $(2,2,2)$ type bow data. The symmetry equations are
    \begin{equation}
        (\varphi\circ g)\cdot \begin{pmatrix}
            g^{-1} I_\pm e^{i(\frac{\theta}{2}\pm\psi)} \\ J_\pm g e^{i(\frac{\theta}{2}\mp\psi)}
        \end{pmatrix}=\begin{pmatrix}
            I_\pm \\ J_\pm
        \end{pmatrix}\quad\Rightarrow\quad \begin{pmatrix}
            \rho_K I_\pm \\ \rho_K J_\pm^\dagger
        \end{pmatrix}=\begin{pmatrix}
            I_\pm e^{i(\frac{\theta}{2}\pm\psi)} \\ J_\pm^\dagger e^{-i(\frac{\theta}{2}\pm\psi)}
        \end{pmatrix}
    \end{equation}
    For $K=0$, we have
    \begin{equation}
        \begin{pmatrix}
            I_\pm \\ J_\pm^\dagger
        \end{pmatrix}=\begin{pmatrix}
            I_\pm e^{i(\frac{\theta}{2}\pm\psi)}\\
            J_\pm^\dagger e^{-i(\frac{\theta}{2}\pm\psi)}
        \end{pmatrix}
    \end{equation}
    Since $(I_\pm,J_\pm^\dagger),(I_-,I_+),(J_-^\dagger,J_+^\dagger)\neq0$, the only values of $\psi$ that allow for solutions are $\psi=\pm\frac{\theta}{2}$. Consider $K=1$, the eigenvalues of $\rho$ are $e^{i\frac{\theta}{2}}, e^{-i\frac{\theta}{2}}$, with eigenvectors $V_1=\begin{pmatrix}
        -i \\ 1
    \end{pmatrix}$, and $V_2=\begin{pmatrix}
        i \\ 1
    \end{pmatrix}$ respectively. Therefore, there are only valid solutions for $\psi=0,\pm\theta$.

    Considering a set of bow data $\mathcal{BD}$ symmetric under an action $\varphi(\theta;\phi;\psi)$, and considering the action \eqref{eq:Reflections}, the bow data $\tilde{\mathcal{BD}}=F\cdot\mathcal{BD}$ generated by the action of $F$ is symmetric under $\varphi(\theta;-\phi;-\psi)$. Therefore, for any bow data symmetric under an action $\varphi\left(\theta;\phi=\frac{\theta}{2};\psi\right)$, we can construct bow data symmetric under the action $\varphi\left(\theta;\phi=-\frac{\theta}{2};-\psi\right)$ by an application of $F$. We therefore only need to consider $\phi=\frac{\theta}{2}$, and all other solutions can be constructed from this.

    We therefore need to consider the cases:
    \begin{itemize}
        \item $(2,2,2)$; $K=0,\psi=\pm\frac{\theta}{2}$: \ref{subsection: Trivial compensating gauge}
        \item $(2,2,2)$; $K=1,\psi=\pm\theta$: \ref{subsection: Non-trivial compensating gauge and psi non-zero}
        \item $(2,2,2)$; $K=1,\psi=0$: \ref{subsection: non-trivial compensating gauge and psi=0}
        \item $(1,2,1)$; $K=0$: \ref{subsection: trivial compensating gauge (1,2,1)}
        \item $(1,2,1)$; $K=1$: \ref{subsection: non-trivial compensating gauge (1,2,1)}
        \item $(2,1,2)$; $K=0$: \ref{subsection: (2,1,2) trivial}
        \item $(2,1,2)$; $K=1$: \ref{subsection: (2,1,2) non-trivial}
    \end{itemize}
    the full parameter spaces of which have been constructed in the relevant sections.
\end{proof}

\section{Limits of the bow data}\label{section: limits of bow data}

The Taub-NUT metric has 2 limits we wish to consider. Recall the metric \eqref{eq:Taub-NUT metric}:
\begin{equation*}
    \eta_\Sigma=\frac{1}{4}\left(V(\vec{t})d\vec{t}^2+\frac{1}{V(\vec{t})}(d\tau+\omega)^2\right)\qquad:\qquad V(\vec{t})=l+\frac{1}{|\vec{t}-\vec{q}|}
\end{equation*}
for $l\in\mathbb{R}_+, \vec{q}\in\mathbb{R}^3$. Considering the limit $l\rightarrow0$, we have the Euclidean metric on $\mathbb{R}^4$. We can see this by writing the flat metric in terms of quaternions \cite{Gibbons1997HyperKahlerSpaces,OHara2021Yang-MillsImpurities}: a quaternion $p\in\mathbb{H}$ can be written as the product of a pure quaternion $a$ and a unit complex number $e^{i\frac{\tau}{2}}$, where $\tau\sim\tau+4\pi$. Defining the vector $\vec{t}:=-aia$, and the 1-form $\omega:=\frac{1}{|\vec{t}|}(da\cdot ia+ai\cdot da)$, the flat metric is
\begin{equation}
    ds^2=dpd\bar{p}=\left(\frac{1}{|\vec{t}|}d\vec{t}^2+|\vec{t}|(d\tau+\omega)^2\right)
\end{equation}
Which is, up to a rescaling, the Taub-NUT metric with $\vec{q}=0$, in the limit $l\rightarrow0$. While, considering the limit $|\vec{q}|\rightarrow\infty$, we have
\begin{equation}\label{eq:R^3xS^1 flat limit of the metric}
    \lim_{\vec{q}\rightarrow\infty}\eta_\Sigma=\frac{1}{4}\left(ld\vec{t}^2+\frac{1}{l}(d\tau+\omega)^2\right)
\end{equation}
which is the flat metric on $\mathbb{R}^3\times S^1$.

% \begin{equation}\label{eq:R^4 flat limit of the metric}
%     \lim_{l\rightarrow0}\eta_\Sigma=\frac{1}{4}\left(\frac{1}{|\vec{t}-\vec{q}|}d\vec{t}^2+|\vec{t}-\vec{q}|(d\tau+\omega)^2\right)
% \end{equation}

It is natural, therefore, to ask the question: does a Taub-NUT instanton recovers an $\mathbb{R}^4$-instanton or a caloron (an anti-self-dual connection on $\mathbb{R}^3\times S^1$) respectively in these limits? From a bow data perspective, we wish to demonstrate that a set of bow data $\mathcal{BD}$ limits to the ADHM data of an instanton when $l\rightarrow0$, and to the Nahm data of a caloron when $|\vec{q}|\rightarrow\infty$.

We will only consider bows of constant rank, and therefore vanishing magnetic charge.

\subsection{The instanton limit}\label{subsection: the instanton limit}

Consider a set of bow data $\mathcal{BD}$ satisfying Nahm's equations, and the bifundamental and fundamental matching conditions. For each of the 3 intervals, $\mathcal{I}_P$, let $l_P$ denote the length of the interval and $\nu_P$ denote its midpoint. We consider the total length of the bow $l_1+l_2+l_3=l\in(0,\varepsilon)$, for some $\varepsilon>0$, and define a normalised interval length $l_P^\prime\in[0,1]$ such that $l_1^\prime+l_2^\prime+l_3^\prime=1$ and $l_p=l\cdot l_P^\prime$. Now consider the Taylor expansion of the Nahm data around the midpoint of each interval
\begin{align}
    &\alpha_P(s)=\alpha_P(\nu_P)+\alpha_P^\prime(\nu_P)(s-\nu_P)+\mathcal{O}(s-\nu_P)^2\label{eq:Taylor expansion of alpha}\\
    &\beta_P(s)=\beta_P(\nu_P)+\beta_P^\prime(\nu_P)(s-\nu_P)+\mathcal{O}(s-\nu_P)^2\label{eQ:Taylor expansion of beta}
\end{align}
the matching conditions therefore become:
\begin{align}
    &A_1(\nu_1)-\frac{1}{2}A_1^\prime(\nu_1)l_1+\mathcal{O}(l_1^2)=(B_+-B_-^\dagger\xi)(B_-+B_+^\dagger\xi)+q \label{eq:Taylor expansion of A1 bifundamental condition}\\
    &A_1(\nu_1)-A_2(\nu_2)+\frac{1}{2}\left(A_1^\prime(\nu_1)l_1+A_2(\nu_2)l_2\right)+\mathcal{O}(l_1^2)+\mathcal{O}(l_2^2)= (I_--J_-^\dagger\xi)(J_-+I_-^\dagger\xi)\label{eq:Taylor expansion of A1A2 fundamental condition}\\
    &A_2(\nu_2)-A_3(\nu_3)+\frac{1}{2}\left(A_2^\prime(\nu_2)l_2+A_3^\prime(\nu_3)l_3\right)+\mathcal{O}(l_2^2)+\mathcal{O}(l_3^2)= (I_+-J_+^\dagger\xi)(J_++I_+^\dagger\xi)\label{eq:Taylor expansion of A2A3 fundamental condition}\\
    &A_3(\nu_3)+\frac{1}{2}A_3^\prime(\nu_3)l_3+\mathcal{O}(l_3^2)=(B_-+B_+^\dagger\xi)(B_+-B_-^\dagger\xi)+q\label{eq:Taylor expansion of A3 bifundamental condition}
\end{align}
Therefore:
\begin{align*}
    &l\cdot\left(A_1(\nu_1)l_1^\prime+A_2(\nu_2)l_2^\prime+A_3(\nu_3)l_3^\prime+\mathcal{O}(l)\right)\\
    &\quad=(B_-+B_+^\dagger\xi)(B_+-B_-^\dagger\xi)-(B_+-B_-^\dagger\xi)(B_-+B_+^\dagger\xi)\\
    &\qquad+(I_--J_-^\dagger\xi)(J_-+I_-^\dagger\xi)+(I_+-J_+^\dagger\xi)(J_++I_+^\dagger\xi)
\end{align*}
So, defining:
\begin{equation}\label{eq:Defining the fundamental ADHM data}
    I:=\begin{pmatrix}
        I_- & I_+
    \end{pmatrix}\qquad;\qquad J^\dagger:=\begin{pmatrix}
        J_-^\dagger & J_+^\dagger
    \end{pmatrix}
\end{equation}
we have:
\begin{align}
    &[B_-,B_+]+IJ=\mathcal{O}(l)\label{eq:Complex ADHM equation for bow data}\\
    &[B_-^\dagger,B_-]+[B_+^\dagger,B_+]+II^\dagger-J^\dagger J=\mathcal{O}(l)\label{eq:Real ADHM equation for bow data}
\end{align}
Now, define the operator:
\begin{equation}\label{eq:Definition of the ADHM operator for bow data}
    \Delta:=\lim_{l\rightarrow0}\begin{pmatrix}
        I^\dagger & J \\ iB_-^\dagger & -iB_+ \\ iB_+^\dagger & iB_-
    \end{pmatrix}
\end{equation}
then, the ADHM equation \cite{Atiyah1978ConstructionInstantons} $\Im_{\mathbb{H}}\left(\Delta^\dagger\Delta\right)=0$ is satisfied if and only if
\begin{align}
    &\lim_{l\rightarrow0}\left(II^\dagger-J^\dagger J+[B_+^\dagger,B_+]+[B_-^\dagger,B_-]\right)=0\label{eq:complex ADHM}\\
    &\lim_{l\rightarrow0}\left(IJ+[B_-,B_+]\right)=0\label{eq:real ADHM}
\end{align}
which is true if and only if the bow conditions are satisfied. Therefore, if we have a set of bow data, solving the bow conditions, and the limit \eqref{eq:Definition of the ADHM operator for bow data} exists, then $\Delta$ is the ADHM operator for an instanton on $\mathbb{R}^4$.

\subsubsection*{(2,2,2) instantons with trivial compensating gauge}

The bow data already solve the ADHM conditions: \eqref{eq:complex ADHM}
\begin{equation*}
    IJ+[B_-,B_+]=I_-J_-+I_+J_++[B_-,B_+]=0
\end{equation*}
is solved trivially since $I_\pm,J_\mp,B_+=0$ for $\psi=\pm\frac{\theta}{2}$, while \eqref{eq:real ADHM} is equal to \eqref{eq:ADHM equation for trivial compensating gauge bow data}, the equation we are left solving in this case. The bow data therefore already define the ADHM data for an instanton, and we can see that the $U(1)$-symmetric ADHM data is recovered.

\subsubsection*{(2,2,2) instantons with non-trivial compensating gauge and non-trivial global phase}

Considering the bow data \ref{subsubsection: summary of the data}, and considering the ADHM equations \eqref{eq:complex ADHM}, \eqref{eq:real ADHM}, we see that:
\begin{equation}\label{eq:ADHM equation for (2,2,2) Nahm data}
    II^\dagger-J^\dagger J+[B_+^\dagger,B_+]+[B_-^\dagger,B_-]=[B_+^\dagger,B_+]=2\kappa_1\tan(\kappa_1(\frac{l}{2}-s_1))\sigma^2
\end{equation}
and
\begin{equation}
    IJ+[B_-,B_+]=I_\pm J_\mp+[B_-,B_+]=0.
\end{equation}
Considering small $l\in(0,\varepsilon)$, we can see that $\kappa_1\tan(\kappa_1(\frac{l}{2}-s_1))$ is small and the remaining data dominates, and therefore in the limit the ADHM equations are solved, with ADHM data:
\begin{equation}
    B_-=Z\operatorname{id},\quad B_+=0,\quad I=\begin{pmatrix}
        W\begin{pmatrix}
            i \\ 1
        \end{pmatrix} & 0
    \end{pmatrix},\quad J=\begin{pmatrix}
        0 \\ W\begin{pmatrix}
            i & 1
        \end{pmatrix}
    \end{pmatrix}
\end{equation}
for $Z,W\in\mathbb{C}$.

\subsubsection*{(2,2,2) instantons with non-trivial compensating gauge and trivial global phase}

Considering the data in \ref{subsection: non-trivial compensating gauge and psi=0}, and considering the ADHM equations \eqref{eq:complex ADHM}, \eqref{eq:real ADHM}, we see that:
\begin{equation}
    II^\dagger-J^\dagger J+[B_+^\dagger,B_+]+[B_-^\dagger,B_-]=[B_+^\dagger,B_+]=2\kappa_1\tan(\kappa_1(\frac{l}{2}-s_1))\sigma^2
\end{equation}
and
\begin{equation}
    IJ+[B_-,B_+]=I_\pm J_\mp+[B_-,B_+]=0
\end{equation}
similar to above. Considering small $l\in(0,\varepsilon)$, we can see that $\kappa_1\tan(\kappa_1(\frac{l}{2}-s_1))$ is small and the remaining data dominates. So, in the limit the ADHM equations are solved with ADHM data
\begin{equation}
    B_-=Z\operatorname{id},\quad B_+=0,\quad I=\begin{pmatrix}
        W_-\begin{pmatrix}
            -i \\ 1
        \end{pmatrix} & W_+\begin{pmatrix}
            -i \\ 1
        \end{pmatrix}
    \end{pmatrix},\quad J=\begin{pmatrix}
        -W_+ \begin{pmatrix}
            1 & i
        \end{pmatrix} \\ W_- \begin{pmatrix}
            1 & i
        \end{pmatrix}
    \end{pmatrix}
\end{equation}
for $Z,W_\pm\in\mathbb{C}$.

% \subsubsection*{Instantons with non-zero magnetic charge}

% Both the $(1,2,1)$ and $(2,1,2)$ bow data limit to ADHM data for a 1-instanton. In general, if the limit exists, the bow data of either a $(k+1,k,k+1)$ or a $(k,k+1,k)$ type bow will recover the ADHM data for a $k$-instanton.

\subsection{The caloron limit}\label{subsection: the caloron limit}

Unlike the instanton limit, we do not have a general argument for the caloron limit of the bow data. Instead, we will approach each case separately. The Nahm data for a caloron is defined over a circle with 2 fundamental points, which split the circle into 2 regions \cite{Charbonneau2010TheCalorons}. To recover caloron Nahm data, therefore, we require that the bow data, in the limit $|\vec{q}|\rightarrow\infty$, becomes periodic across the bifundamental point. In which case the bifundamental point can be disregarded and we are left with periodic Nahm data over a circle with 2 fundamental points, in other words we are left with the Nahm data of a caloron. In the case that the limit exists, we will be comparing the produced data with the calorons given by Harland \cite{Harland2007LargeCalorons}, and we show that we can recover all calorons presented in that paper. We note that these limits have been taken in a gauge in which $T^0=0$ across the whole bow, this requires a non-periodic gauge transformation. A constant $T^0$ may be recovered for caloron Nahm data using a non-periodic gauge transformation in the centraliser of the compensating gauge.

\subsubsection*{(2,2,2) instantons with trivial compensating gauge}

Here $\alpha_1=\alpha_3=\frac{1}{2}(2q-|Z|^2-|W|^2)\operatorname{id}$, and $\beta_1=\beta_3=0$ are already periodic, so considering some limit of the bifundamental data such that
\begin{equation}
    \lim_{|q|\rightarrow\infty}\left(2q-|Z|^2-|W|^2\right)=L<\infty
\end{equation}
leaves the Nahm data finite and periodic, and leaves the fundamental data finite. This therefore produces caloron Nahm data and recovers the trivial $U(1)$-symmetric $(2,2)$ caloron Nahm data.

\subsubsection*{(2,2,2) instantons with non-trivial compensating gauge and non-trivial global phase}

We need to take the $|q|\rightarrow \infty$ limit in such a way that
\begin{equation}
    \lim_{|q|\rightarrow\infty} u_1=\lim_{|q|\rightarrow\infty}\left(q+\frac{1}{2}\kappa_1\cot\left(2\kappa_1\left(\frac{l}{2}-s_1\right)\right)\right)=L<\infty
\end{equation}
in order to do this, write $q=L-\frac{1}{2(l-2s_1)}+\mathcal{O}(\frac{l}{2}-s_1)$ for fixed $L<\infty$. Then $|q|\rightarrow\infty$ as $s_1\rightarrow\frac{l}{2}^-$, and from the Laurent expansion of the cotangent we have
\begin{equation}
    \lim_{s_1\rightarrow\frac{l}{2}^-}\left(q+\frac{1}{2}\kappa_1\cot\left(2\kappa_1\left(\frac{l}{2}-s_1\right)\right)\right)=\lim_{s_1\rightarrow\frac{l}{2}^{-}}\left(L+\mathcal{O}\left(\frac{l}{2}-s_1\right)\right)=L.
\end{equation}
With this, we can see that $B_+=\sqrt{\kappa_1\csc(\kappa_1(\frac{l}{2}-s_1))}e^{i\eta}\operatorname{id}\rightarrow\infty$, and $B_-=\sqrt{\frac{1}{2}\kappa_1\tan(\kappa_1(\frac{l}{2}-s_1))}(\sigma^3+i\sigma_1)\rightarrow0$, and the remaining data is finite with $\alpha_1(-\frac{l}{2})=\alpha_3(\frac{l}{2})$, and $\beta_1(-\frac{l}{2})=\beta_1(\frac{l}{2})$. We can therefore see that the bow data becomes periodic, and so recovers the non-trivial $U(1)$-symmetric $(2,2)$ caloron.

\subsubsection*{(2,2,2) instantons with non-trivial compensating gauge and trivial global phase}

Since the bifundamental matching conditions are the same in this case as the previous one, we again require that $\lim_{|q|\rightarrow\infty}u_1=L<\infty$, and therefore we wish to follow the same analysis. In the above analysis we required taking the limit $s_1\rightarrow\frac{l}{2}^-$, however recall that for $\psi=0$ we required $s_1<\lambda<\frac{l}{2}$ to ensure that $l_3=\kappa_1\tan(\kappa_1(\lambda-s_1))-\kappa_2\tan(\kappa_2\lambda)>0$ was possible. We can therefore see that the caloron limit cannot exist for this class of solutions to the bow. This matches \cite{Harland2007LargeCalorons}, in which there are no caloron solutions to the Nahm data with a trivial global phase. This demonstrates the existence of Taub-NUT instantons that do not limit to calorons under $|q|\rightarrow\infty$.

% \subsubsection*{Instantons with non-zero magnetic charge}

% Both the $(1,2,1)$ and $(2,1,2)$ instantons limit to $(2,1)$ calorons, the Nahm data of which are gauge equivalent up to the action of the `rotation map' which interchanges the roles of the constituent monopoles in the caloron, in turn interchanging the ranks of the intervals. In the trivial compensating gauge cases, the bow data already defines caloron Nahm data since it is already periodic across the interval. The non-trivial compensating gauge $(1,2,1)$ bow data is also already periodic, and therefore defines the Nahm data of a caloron; in the $(2,1,2)$ case, however, a similar analysis as in the $(2,2,2)$ case can be made to ensure the bow data is periodic across the point $s=\frac{l}{2}$. In these limits the trivial and non-trivial $(2,1)$ caloron Nahm data constructed by Harland is recovered.

\section{Conclusions}

We have presented the first systematic analytic solutions of the bow, as presented in \cite{Cherkis2021InstantonsTheorem,Cherkis2024InstantonsConstruction, Cherkis2025InstantonsIsometry}, describing the moduli spaces of Taub-NUT instantons beyond charge 1. We have done this by constructing and classifying all $U(1)$-symmetric Taub-NUT instanton moduli spaces for $(k,m)=(2,0),(1,1),(2,-1)$. In doing so we have demonstrated that most of our constructed solutions can be directly related to $\mathbb{R}^4$-instantons or to calorons in particular limits of the bow; we have also demonstrated, however, the existence of Taub-NUT instantons \ref{subsection: non-trivial compensating gauge and psi=0} that do not limit to calorons.

One potential direction of future work is to investigate the moduli spaces constructed in more detail, to consider their induced metrics and their properties. The moduli spaces constructed are geodesic submanifolds of the full moduli space since the action considered acts as an isometry. The symmetry actions considered fix exactly one of the complex structures of the moduli space, and as such the symmetric moduli spaces will be K\"{a}hler submanifolds of the full moduli space, this is corroborated by the fact all the constructed moduli spaces are even dimensional. Therefore, a more in-depth study of the symmetric moduli space could give insights into the structure of the full moduli space.

In particular, a more detailed study of the instantons with trivial phase action would be of interest, since these do not limit to calorons and represent a class unique to the Taub-NUT space. A numerical construction of the gauge field on the Taub-NUT space through the bow transform could demonstrate the way in which these are different from the examples that do produce calorons.

There is also the possibility of extending these results to a classification of instantons on the multi-Taub-NUT spaces; the bow transform generalises in a somewhat straightforward way to ALF gravitons with $N$ NUT centres, so a similar classification for multi-Taub-NUT may be possible. In this case, considering the limit $l\rightarrow0$ gives an $(N-1)$-centred ALE space \cite{Kronheimer1989TheQuotients}, and we would expect the corresponding bow to limit to the ADHM data for an instanton on said ALE space \cite{Kronheimer1990Yang-MillsInstantons}. Alternatively, considering the limit where one centre is sent to spatial infinity the $N$-centred Taub-NUT space becomes an $(N-1)$-centred Taub-NUT space, here we would expect the effect on the bow data to be that of removing a bifundamental point, in a similar way to removing the bifundamental point from the single centred bow leaves the Nahm data for a caloron.

\subsubsection*{Acknowledgements}

I would like to thank my PhD supervisor Dr Josh Cork for the many useful discussions and insights.

\printbibliography

\end{document}